\documentclass[a4paper,11pt]{article}
\usepackage{jheppub} 

\usepackage{jheppub}
\usepackage{amsfonts}
\usepackage{amsmath}
\usepackage{amssymb}
\usepackage{graphicx}
\usepackage{amsthm}
\usepackage[mathscr]{eucal}
\usepackage{bbold}
\usepackage{braket}
\usepackage{color}
\usepackage{colortbl} 
\usepackage{array} 
\usepackage{dsfont}
\usepackage{framed}
\usepackage{mathtools}
\usepackage{physics}
\usepackage[normalem]{ulem}
\usepackage{tensor}
\usepackage{thmtools}
\usepackage{thm-restate}
\usepackage{ulem}
\usepackage{tikz}
\usetikzlibrary{arrows.meta,arrows} 
\usepackage{centernot}
\usepackage{enumitem} 
\usepackage[a]{esvect}  
\usepackage{slashed}

\usepackage{yfonts}
\usepackage{float}
\usetikzlibrary{math} 
\usetikzlibrary{calc}		

\usepackage{subcaption}
\usepackage{float}
\usepackage{afterpage}
\renewcommand{\thefigure}{\arabic{figure}}

\usetikzlibrary{arrows}
\tikzset{elabelcolor/.style={color=blue} 
    vertex/.style={circle,draw,minimum size=1.5em},
    edge/.style={->,> = latex'}
}  
\usetikzlibrary{decorations.pathreplacing}

\usepackage{cleveref}

\usepackage[ruled,vlined]{algorithm2e}

\definecolor{THcolor}{rgb}{0.1,0.7,0.1}

\definecolor{VHcolor}{rgb}{0.7,0.3,0.9}

\definecolor{MRocolor}{rgb}{0.1,0,1}

\definecolor{gcolor}{RGB}{200,255,180}
\definecolor{ycolor}{RGB}{250,250,180}
\definecolor{rcolor}{RGB}{255,200,180}
\definecolor{bcolor}{RGB}{210,210,255}
\definecolor{llg}{RGB}{240,240,240}
\newcolumntype{g}{>{\columncolor{gcolor}}c}
\newcolumntype{y}{>{\columncolor{ycolor}}c}
\newcolumntype{r}{>{\columncolor{rcolor}}c}
\newcolumntype{b}{>{\columncolor{bcolor}}c}

\definecolor{colormacro}{rgb}{0.6,0.4,0.4}  
\definecolor{colormacromath}{rgb}{0.2,0.4,0.8}

\newtheorem{condition}{Condition}

\newtheorem{defi}{Definition}

\newtheorem{fcn}{Function}

\newtheorem{thm}{Theorem}

\newcommand{\N}{{\sf N}}
\newcommand{\psys}[1]{[\![#1]\!]}
\newcommand{\D}{\sf{D}}

\newcommand{\J}{\mathscr{J}}
\newcommand{\K}{\mathscr{K}}
\newcommand{\uI}{\underline{\mathscr{I}}}
\newcommand{\uJ}{\underline{\mathscr{J}}}
\newcommand{\uK}{\underline{\mathscr{K}}}
\newcommand{\ent}{{\sf S}}  
\newcommand{\mi}{{\sf I}}  
\newcommand{\face}{\mathscr{F}}   
\newcommand{\hyp}{\mathbb{H}}  

\newcommand{\mgs}{\mathcal{E}}  

\newcommand{\van}{\mathcal{V}}
\newcommand{\excluded}{\mathcal{U}^{(0)}}

\newcommand{\p}{\mathcal{P}} 
\newcommand{\lat}[1]{\mathcal{L}_{#1}} 
\newcommand{\ineq}{\mathcal{E}} 
\newcommand{\E}{{\sf{E}}} 
\newcommand{\cone}{\mathcal{C}}  
\newcommand{\regcone}{\widehat{\mathcal{C}}}  
\newcommand{\regface}{\widehat{\face}}  
\newcommand{\x}{\mathcal{X}} 
\newcommand{\y}{\mathcal{Y}} 
\newcommand{\ds}{\mathcal{D}} 
\newcommand{\irr}{\mathcal{U}} 
\newcommand{\orb}{\mathcal{M}}  
\newcommand{\stab}[1]{G^{{#1}}} 

\newcommand{\cl}{\text{cl}}
\newcommand{\lc}{\text{cl}_{\text{L}}}
\newcommand{\dc}{\text{cl}_{\text{D}}}
\newcommand{\fc}{\text{cl}_{\text{F}}}
\newcommand{\ldc}{\text{cl}_{\text{LD}}}
\newcommand{\fdc}{\text{cl}_{\text{FD}}}

\newcommand{\erstar}{\mathscr{R}_\star} 
\newcommand{\erst}{\mathscr{R}_{_\text{STG}}} 
\newcommand{\ert}{\mathscr{R}_{_\text{TG}}} 
\newcommand{\erh}{\mathscr{R}_{_\text{Hol}}} 
\newcommand{\erhg}{\mathscr{R}_{_\text{HyG}}} 
\newcommand{\ersta}{\mathscr{R}_{_\text{Sta}}} 
\newcommand{\erq}{\mathscr{R}_{_\text{QM}}} 
\newcommand{\erssa}{\mathscr{R}_{_\text{SSA}}} 
\newcommand{\erkc}{\mathscr{R}_{_\text{KC}}} 
\newcommand{\ersa}{\mathscr{R}_{_\text{SA}}} 

\SetKwComment{Comment}{** }{ **} 
\SetKwInOut{Input}{input}\SetKwInOut{Output}{output}

\title{Algorithmic construction of SSA-compatible extreme rays of the subadditivity cone and the $\N=6$ solution}

\author[a]{Temple He,}
\emailAdd{templehe@caltech.edu}

\author[b]{Veronika E. Hubeny,}
\emailAdd{veronika@physics.ucdavis.edu}

\author[b,c]{Massimiliano Rota}
\emailAdd{max.rota@bristol.ac.uk}

\affiliation[a]{Walter Burke Institute for Theoretical Physics \\ California Institute of Technology, Pasadena, CA 91125 USA}
\affiliation[b]{Center for Quantum Mathematics and Physics (QMAP)\\ 
Department of Physics \& Astronomy, University of California, Davis, CA 95616 USA}
\affiliation[c]{School of Mathematics, University of Bristol,
Woodland Road, Bristol, BS8 1UG UK}

\abstract{We compute the set of all extreme rays of the 6-party subadditivity cone that are compatible with strong subadditivity. In total, we identify 208 new (genuine 6-party) orbits, 52 of which violate at least one known holographic entropy inequality. For the remaining 156 orbits, which do not violate any such inequalities, we construct holographic graph models for 150 of them. For the final 6 orbits, it remains an open question whether they are holographic. Consistent with the strong form of the conjecture in \cite{Hernandez-Cuenca:2022pst}, 148 of these graph models are trees. However, 2 of the graphs contain a ``bulk cycle'', leaving open the question of whether equivalent models with tree topology exist, or if these extreme rays are counterexamples to the conjecture. The paper includes a detailed description of the algorithm used for the computation, which is presented in a general framework and can be applied to any situation involving a polyhedral cone defined by a set of linear inequalities and a partial order among them to find extreme rays corresponding to down-sets in this poset.}

\begin{document}
\hfill{CALT-TH 2024-049}

\maketitle

\section{Introduction}
\label{sec:intro}

The study of correlations is of central importance in quantum information theory, as it provides a comprehensive understanding of how information is distributed among different parts of a quantum system.  Correlations are often explored through various measures and related constraints, which help in understanding the fundamental limits and capabilities of quantum systems, guiding the development of quantum technologies, and deepening our theoretical understanding of quantum mechanics. This work will focus on constraints which take the form of linear inequalities involving the von Neumann entropy of subsystems of a multipartite quantum system.

While some of these inequalities are ``universal'', in the sense that they are obeyed by arbitrary quantum states, others are specific to particular classes of states, like stabilizers states \cite{Linden:2013kal,gross-walter}, or classical probability distributions \cite{641556,1998:Zhang,2002:Makarychev,2003:Zhang,2007:Matus}. These ``class dependent'' inequalities therefore serve as characterizations of the states in these classes.

A class which is of particular interest for this work is that of ``geometric states'' in quantum gravity, or more specifically, in the context of gauge/gravity duality \cite{Maldacena:1997re,Gubser:1998bc,Witten:1998qj}. These are the states in holographic CFTs which correspond to bulk classical geometries, for which the entropies of boundary spatial subsystems can be computed via the HRRT prescription \cite{Ryu:2006bv,Hubeny:2007xt}. The fact that these states obey entropy inequalities that are not universal was first noticed in \cite{Hayden:2011ag}, and then explored more systematically in \cite{Bao:2015bfa}. Since then, many new inequalities have been proven (for instance, see \cite{Czech:2022fzb,Hernandez-Cuenca:2023iqh,Czech:2024rco}) using the ``contraction map'' method of \cite{Bao:2015bfa}, but a deeper conceptual understanding of the structure of the resulting \textit{holographic entropy cone} (HEC) \cite{Bao:2015bfa} for an arbitrary number of parties is still lacking. It is reasonable to expect that such an understanding would provide new insights about the entanglement structure of geometric states, and ultimately shed new light on how the bulk geometry is ``encoded'' in the boundary theory.

Motivated by this problem, and based on the ideas from \cite{Hubeny:2018trv,Hubeny:2018ijt}, a deep question about fundamental possible ``patterns of dependences'' in quantum mechanics was raised in \cite{Hernandez-Cuenca:2019jpv}, which dubbed it the \textit{quantum marginal independence problem} (QMIP). In its simplest and most intuitive form, the QMIP asks the following question: Given the binary specification of presence versus absence of correlations between all pairs of disjoint subsystems, does there exist any quantum state realizing such a choice?
While this problem was formulated in \cite{Hernandez-Cuenca:2019jpv} for arbitrary quantum states, and it is seemingly unrelated to questions concerning the structure of the HEC, it was argued in \cite{Hernandez-Cuenca:2022pst} that the solution to the QMIP within the restricted class of geometric states (the \textit{holographic marginal independence problem} (HMIP)) would provide the ``essential data'' from which the HEC can be reconstructed. In fact, \cite{Hernandez-Cuenca:2022pst} even argued that it would be sufficient to only know the answer to the ``extremal'' version of the HMIP, where one only asks what are the possible situations where the number of independent subsystems is so large that demanding any additional independence would force the state to completely factorize. Formally, as we will review below, answering this question requires one to determine which extreme rays of the ``subadditivity cone'' can be realized by geometric states, or in practical terms, by the graph models of \cite{Bao:2015bfa}.

Since the HEC is known explicitly for any number of parties $\N\leq5$ \cite{Bao:2015bfa,Cuenca:2019uzx}, the conjectures in \cite{Hernandez-Cuenca:2022pst} can easily be verified to be true in this context. However, while this is already a rather non-trivial check, it was previously observed in \cite{He:2023cco} (using preliminary data from the present work) that the $\N=6$ case might have a much richer structure, and could therefore not only serve to provide additional highly non-trivial checks, but also to obtain more data that could guide intuition towards a characterization of the extremal marginal independence patterns which are holographic. With this as our main motivation, we develop in this work an algorithm to compute all the extreme rays of the $\N=6$ subadditivity cone that can possibly be realized by quantum states (not just geometric states in holography), or more precisely, all extreme rays that satisfy strong subadditivity. Below, we will review the main definitions and terminology, comment more technically on our motivations beyond the holographic framework, and outline the strategy that we will employ for the computation.

\paragraph{Formulation of the problem:}
Let us label the parties $1,2,\ldots\N$, and denote the purifier $0$, and let $\uJ,\uK$ denote an arbitrary pair of disjoint subsystems, i.e., $\uJ,\uK\subseteq \psys{\N}=\{0,1,\ldots,\N\}$ with $\uJ\cap\uK=\varnothing$. Suppose that for each such pair we only specify whether the two subsystems are correlated or not, without any further specification about the amount or other properties of such correlations. Under what conditions does there exist a \textit{pure} quantum state such that all these demands are satisfied?

As explained in \cite{Hernandez-Cuenca:2019jpv}, and then further elaborated in \cite{He:2022bmi}, 
one can conveniently avoid unnecessary consideration of many obviously unrealizable patterns by reformulating the QMIP in terms of faces of the \textit{subadditivity cone} (SAC). 
To define the SAC, one works in the space of subsystem von Neumann entropies. For an arbitrary $\N$-party density matrix $\rho_{\N}$, one first defines its \textit{entropy vector} to be 
\begin{equation}
\label{eq:entropy_vector}
    \vec{\ent}(\rho_{\N})=\{\ent_{\J},\; \forall\,\J\;\;\text{such that}\;\; \varnothing\neq\J\subseteq [\N]\} ,
\end{equation}
where components are ordered according some fixed convention (e.g., lexicographically), and $\ent_{\J}$ is the von Neumann entropy of the reduced density matrix on the $\J$ subsystem (notice the distinction between indices $\uJ\subseteq\psys{\N}$ and $\J\subseteq [\N]$). The vector space $\mathbb{R}^{\D}$ with $\D=2^{\N}-1$, where this vector lives, will be called \textit{$\N$-party entropy space}.\footnote{\,We will call an arbitrary element of this vector space an entropy vector, even if it is not necessarily the vector of entropies of a density matrix as defined in \eqref{eq:entropy_vector}.} The $\N$-party subadditivity cone SAC$_{\N}$ is then defined as the \textit{polyhedral} cone in $\N$-party entropy space specified by all instances of \textit{subadditivity} (SA), namely for any two disjoint nonempty subsystems $\uJ,\uK$,
\begin{equation}
\label{eq:sa}
    \ent_{\uJ}+\ent_{\uK}\geq \ent_{\uJ\uK} ,
\end{equation}
where we used the conventional compact notation $\ent_{\uJ\uK}$ for $\ent_{\uJ\cup\uK}$, and if $\uJ$ contains the purifier $0$, we implicitly replace $\ent_{\uJ}$ with the equivalent $\ent_{\uJ^c}$ (and likewise for $\uK$ and $\uJ\uK$).\footnote{\,We also implicitly assume $\ent_{\varnothing}$=0.}

Since \eqref{eq:sa} vanishes in quantum mechanics if and only if $\rho_{\uJ\uK} = \rho_{\uJ} \otimes \rho_{\uK}$, which means there are no correlations between $\uJ$ and $\uK$, the pattern of (in)dependences among the subsystems of an arbitrary purification of $\rho_{\N}$ (or equivalently, an arbitrary quantum state involving $\N+1$ parties) is completely captured by the unique \textit{face} $\face$ of the SAC$_{\N}$ for which 
\begin{equation}
\label{eq:pmi}
    \vec{\ent}(\rho_{\N})\in \text{int}(\face) ,
\end{equation}
with $\text{int}(\face)$ denoting the interior of $\face$ according to the standard topology of $\mathbb{R}^{\D}$.\footnote{\,Since any face $\face$ is specified by a set of SA instances which are saturated, if $\vec{\ent}(\rho_{\N})\in \partial\face$, then there is a lower-dimensional face $\face'\subset\face$ such that $\vec{\ent}(\rho_{\N})\in \text{int}(\face')$.
} 
The QMIP then asks for which faces of the SAC$_{\N}$ does there exist a density matrix $\rho_{\N}$ such that \eqref{eq:pmi} holds. A face for which such a density matrix exists will be said to be \textit{realizable}.\footnote{\,In some situations, it might be useful to slightly relax this definition of realizability and allow for the possibility that an entropy vector in the interior of a face is only approximated arbitrarily well by a density matrix (see e.g., \cite{He:2023aif}). This distinction however will not play a role here.}

\paragraph{Context and motivation:}
This formulation of the QMIP makes it obvious that the problem can easily be solved if one knows all the $\N$-party entropy inequalities. Indeed, for any given face $\face$ of the SAC$_{\N}$, one simply needs to check if there is at least one entropy vector in the interior of $\face$ that satisfies all inequalities. A more interesting perspective instead is to interpret the QMIP as a more fundamental problem in quantum mechanics that could be solved by other means, and then ask what are the implications of its solution for the existence of new entropy inequalities, and what form they might have. If such solution is known up to some $\N$, it is easy to see that the information it carries about $\N$-party inequalities is very limited. This seems to change dramatically however, at least in some restricted scenarios, if one instead asks how much information about the $\N$-party entropy inequalities is encoded in the solution to the QMIP for \textit{all} $\N'\geq\N$ (or possibly for all $\N'\le f(\N)$ for some function $f$).

In the holographic context, it was argued in \cite{Hernandez-Cuenca:2022pst} that, for any given number of parties $\N$, if one knows the solution to the HMIP \textit{for each} $\N'\geq\N$, then one can in principle reconstruct all the inequalities that specify the $\N$-party holographic entropy cone. In fact, as we mentioned above, \cite{Hernandez-Cuenca:2022pst} suggested that there might be an even more remarkable structure which only requires the solution to the extremal version of the HMIP rather than the full one.
More precisely, let $\ersa^{^{\N}}$ be the set of all $1$-dimensional faces, or \textit{extreme rays} (ERs), of the SAC$_{\N}$, and $\erh^{^{\N}}\subseteq\ersa^{^{\N}}$ be the subset of these that can be realized holographically.\footnote{\,Similarly, we will use $\mathscr{R}^{^{\N}}$ with other subscripts to indicate subsets of the SAC$_{\N}$ extreme rays pertaining to the class indicated by the subscript; refer to \Cref{tab:R_def} for a summary of these.} Then \cite{Hernandez-Cuenca:2022pst} argued that for any $\N$, there exists some $\N'\geq\N$ (which depends on $\N$) such that $\erh^{^{\N'}}$ provides sufficient information to derive the extreme rays of the $\N$-party holographic entropy cone via a simple coarse-graining of subsystems. These then allow for a straightforward reconstruction of all the $\N$-party holographic entropy inequalities (although in practice this might be a hard computation). The result of \cite{Hernandez-Cuenca:2022pst} is quite striking since the SAC itself is a much more primitive construct than the holographic entropy cone. It originates from the self-evident statement that the amount of correlation between any pair of subsystems cannot be negative, and as such it is not cognizant of holography. Indeed, the holographic input enters through the restriction to $\erh^{^{\N}}$.

Interestingly, even though \cite{Hernandez-Cuenca:2022pst} relied on technology only applicable in the holographic setting, a clue hinting at a similar structure in more general contexts was found recently in \cite{He:2023aif}, which considered how to derive inner bounds to the general \textit{quantum entropy cone} \cite{1193790,2005:Linden,2011:Cadney,Christandl_2023} from the solution to the QMIP. Denoting by $\ersta^{^{\N}}\subseteq\ersa^{^{\N}}$ the set of extreme rays of the SAC$_{\N}$ that can be realized by \textit{stabilizer states}, the main result of \cite{He:2023aif} was that a subset of $\ersta^{^9}$ is sufficient to derive\footnote{\,As in the holographic case, by derive we mean via subsystems coarse-grainings.} an inner bound to the $4$-party quantum entropy cone
which coincides with the full 4-party \emph{stabilizer cone} \cite{Linden:2013kal}.

These results suggest that a similar structure might also exist at a more abstract level. For $\N\geq 3$, most faces of the SAC$_{\N}$ cannot be part of the solution to the QMIP for the simple reason that they are ``ruled out'' by \textit{strong subadditivity} (SSA). Following \cite{He:2022bmi}, we will say that a face is SSA-\textit{compatible} if there is at least one entropy vector in its interior that satisfies all instances of SSA. For a fixed $\N$, the set of SSA-compatible faces has the structure of a lattice,\footnote{\,A lattice is a poset $\p$ endowed with the additional property that for any two elements $x,y \in \p$, there exists both a least upper bound (called the join and denoted $x \vee y$)
and a greatest lower bound (called the meet and denoted $x \wedge y$).}  where the partial order is given by reverse inclusion ($\face_1<\face_2$ if $\face_1\supset\face_2$), and the meet of two faces is the lowest dimensional face $\face=\face_1\wedge\face_2$ that includes $\face_1$ and $\face_2$ on its boundary \cite{He:2022bmi}.\footnote{\,The set of faces of an arbitrary polyhedral cone forms a lattice with the same partial order and meet operation.} In general, the subset $\erssa^{^{\N}} \subseteq \ersa^{^{\N}}$ of SSA-compatible extreme rays of the SAC$_{\N}$ is a subset of the set of coatoms of this lattice,\footnote{\,For $\N\leq 4$, each coatom is an extreme ray, but it is not known if the inclusion is strict for larger $\N$; see \cite{He:2022bmi} for more details.} but \cite{He:2022bmi} showed that for $\N\geq 4$ the lattice is not coatomistic. This means that not every element of the lattice can be obtained as the meet of a collection of coatoms. In other words, even geometrically, knowing $\erssa^{^{\N}}$ is in general insufficient information to be able to recover all the SSA-compatible higher-dimensional faces.\footnote{\,Technically, this result was proven for the set of SSA-compatible faces, which are not necessarily all realizable by quantum states. Therefore, it could have been in principle possible that each element of the lattice which is not the meet of a collection of coatoms is non-realizable, implying that the sublattice of realizable elements is coatomistic. However, we will see momentarily that already for $\N=4$ this is not true.
} This result strongly suggests that in general, the solution to the QMIP for fixed $\N$ cannot be deduced from $\erq^{^{\N}}$, which is only the solution to the extremal version of QMIP involving the same $\N$, and that one needs to know all the meet-irreducible elements of the lattice.\footnote{\,An element $x$ of a lattice is said to be meet-irreducible if $x=y\wedge z$ implies $x=y$ or $x=z$.} For $\N=4$ however, it was already shown in \cite{Hernandez-Cuenca:2019jpv} that all SSA-compatible faces of the SAC$_4$ are realizable in quantum mechanics, and that in fact stabilizer states are sufficient. This implies that the meet-irreducible elements of the lattice of SSA-compatible faces of the SAC$_4$ which are not coatoms can be obtained from coarse-grainings of certain coatoms of the analogous lattice for $\N=9$. In other words, similar to the holographic case, the solution to the QMIP can be derived from the solution to its extremal version involving a larger number of parties. All these observations suggest that there might be classes of quantum states for which the solution to the extremal marginal independence problem for all $\N'\geq\N$ carries a lot of (if not complete)  information about the solution to the (seemingly) more general $\N$-party QMIP, and may even contain information about entropy inequalities. This motivates us to study the set $\erssa^{^{\N}}$ more closely.

Another important motivation for us to study $\erssa^{^{\N}}$ comes again from holography. Suppose that the conjecture proposed in \cite{Hernandez-Cuenca:2022pst} is true, and that all $\N$-party holographic entropy inequalities can indeed be recovered from $\erh^{^{\N'}}$ for some $\N'\geq \N$. Then what characterizes the extreme rays that can be realized in quantum mechanics but not by geometric states? A first guess, informed by data for $\N\leq 5$, is that no such rays exist and $\erq^{^{\N}}=\erh^{^{\N}}$. In other words, all extreme rays of the SAC$_{\N}$ that can be realized in quantum mechanics would also be realizable holographically. However, this was shown to be false in \cite{He:2023cco}, which proved the strict inclusion $\erh^{^{\N}}\subset\erq^{^{\N}}$ for $\N\geq 6$ via the construction of an explicit counterexample, namely an extreme ray of the SAC$_6$ that can be realized by a stabilizer state but not by a geometric state.

To explore all these questions, it is useful to compute $\erssa^{^{\N}}$ for larger $\N$, so that we have more data available for further investigations. In principle, the most straightforward way to derive $\erssa^{^{\N}}$ is to first find $\ersa^{^{\N}}$, and then check which extreme rays satisfy SSA. This strategy however quickly runs into problems as $\N$ grows. For any given $\N$, the number of facets of the SAC$_{\N}$ is given by
\begin{equation}
\label{eq:facets_number}
  \text{number of facets} \,= {\N+2 \brace 3} - 2^{\N} +1 ,
\end{equation}
where ${n \brace k}$ is the Stirling number of the second kind \cite{Hubeny:2018ijt}.\footnote{\,The number of SA instances is ${\N+2 \brace 3}$, but the $2^{\N} -1$ instances of the form $\ent_{\J}\geq 0$ are not facets \cite{He:2022bmi}.} For $\N\leq 4$ the problem is easily manageable by freely available software \cite{Normaliz:301}, but the $\N=5$ case is already significantly more involved, since ${\D} =31$ and there are $270$ facets. The computation can be largely simplified using a result of \cite{He:2022bmi}, which identified for arbitrary $\N$ a specific lower-dimensional face $\face_*$ of the SAC$_{\N}$ that includes all elements of $\erssa^{^{\N}}$ (with the exception of the trivial ones corresponding to the Bell pairs).\footnote{\,Since we will utilize this result in the computation presented in this paper, we will review it in detail in the next section.} In general, the dimension of $\face_*$ is given by
\begin{equation}
\label{eq:special_face_dim}
    \text{dim}(\face_*)=2^{\N}-\binom{\N+1}{2}-1 ,
\end{equation}
which for $\N=5$ gives $\text{dim}(\face_*)=16$. Furthermore, upon dimensional reduction of the inequalities that specify the SAC$_5$ to the linear subspace spanned by $\face_*$, many inequalities become redundant, and $\face_*$ only has $80$ facets. In this manner, the computation of $\ersa^{^5}$ becomes straightforward using \cite{Normaliz:301}, but as one can immediately see from \eqref{eq:facets_number} and \eqref{eq:special_face_dim}, this simplification quickly becomes insufficient as $\N$ grows. Indeed, already at $\N=6$, the face $\face_*$ has $\text{dim}(\face_*)=42$ and is specified by $280$ inequalities (the number of facets of SAC$_{6}$ is $903$), rendering the computation of $\ersa^{^6}$ unfeasible on a standard laptop. The set $\ersa^{^7}$ is certainly beyond reach, and the development of an algorithm that allows us to circumvent at least some of these difficulties will be the main goal of this paper.

\paragraph{Strategy:}
A simple observation, which is key to the development of this algorithm, is that already for small values of $\N$ the vast majority of extreme rays of the SAC$_{\N}$ violate SSA, i.e., $|\erssa^{^{\N}}|\ll|\ersa^{^{\N}}|$.\footnote{\,This is evident already 
at $\N=3,4$ (e.g., see \cite{He:2022bmi}).} Since we are only interested in $\erssa^{^{\N}}$, one may then wonder whether there is a way to find this set more directly, without first deriving $\ersa^{^{\N}}$. A possible strategy would be to first find all the extreme rays of the ``SSA-cone'', the polyhedral cone in entropy space carved out by all instances of SA \textit{and} SSA. It would then appear 
straightforward to extract the subset of extreme rays which are also elements of $\ersa^{^{\N}}$. This strategy unfortunately suffers from the same issues described above, since the dimension of the ambient space is the same, the number of facets is comparable, and most extreme rays of this cone are not elements of $\erssa^{^{\N}}$.

To focus on $\erssa^{^{\N}}$ without having to introduce new inequalities, we will instead consider a \textit{variation} of what is perhaps the most basic (and in general 
inefficient) strategy to find the extreme rays of any polyhedral cone. Consider a vector space $\mathbb{R}^n$ for some $n$, and a finite set $\ineq$ of linear inequalities that specify a full-dimensional pointed\footnote{\,A cone is pointed if it does not contain any linear subspace other than the trivial subspace.} cone $\cone$. The set of vectors that saturate an inequality is a hyperplane in $\mathbb{R}^n$. An arbitrary subset $\x\subseteq\ineq$ then determines a linear subspace of $\mathbb{R}^n$ given by the intersection of the hyperplanes in $\x$. If this subspace is $1$-dimensional, one can then check if it contains a ray that satisfies all the inequalities in $\ineq$. If such a ray exists, it is obviously an extreme ray of $\cone$. By considering all possible subsets $\x\in 2^{\ineq}$, we can in principle find all the extreme rays. We call this strategy the ``brute force strategy.''

The inefficiency of this procedure should be obvious, but what if one is only interested in a subset of extreme rays and can drastically reduce the search space from $2^{\ineq}$ to a much smaller subset? The essence of our method will be to apply the brute force strategy to a sufficiently reduced search space. As we will review in detail, such a reduction can be attained using the strategy of \cite{He:2022bmi}, where one replaces SSA with a weaker condition, dubbed \textit{Klein's condition} (KC), that can be conveniently phrased in terms of collections of SA instances. Following \cite{He:2022bmi}, we will formulate KC by introducing a partial order among the SA instances. The set $\erkc^{^{\N}}$ of extreme rays of the SAC$_{\N}$ that obey this condition will be the set of extreme rays for which the corresponding sets of saturated SA instances are \textit{down-sets}. The algorithm will take advantage of this structural property of extreme rays to restrict the search space from $2^{\ineq}$ to the lattice of down-sets of the poset. This is still far from being a sufficient reduction for $\N\geq 6$, but using \textit{closure operators} we will then ``combine'' this constraint with the dependences among the SA instances to reduce the search space even further. 

Our algorithm will compute $\erkc^{^{\N}}$, and since KC is an approximation to SSA, we have
\begin{equation}
    \erssa^{^{\N}}\subseteq\erkc^{^{\N}}\subseteq\ersa^{^{\N}}.
\end{equation}
To find $\erssa^{^{\N}}$, it then suffices to check which elements of $\erkc^{^{\N}}$ satisfy SSA, which is a straightforward computation. The advantage of this strategy ultimately lies in the expectation that as $\N$ grows,
\begin{equation}
    \frac{|\erkc^{^{\N}}|}{|\erssa^{^{\N}}|}  \ll \frac{|\ersa^{^{\N}}|}{|\erkc^{^{\N}}| }  .
\end{equation}
As verified in \cite{He:2022bmi}, for $\N\leq 5$ we have $\erkc^{^{\N}}=\erssa^{^{\N}}$ (in contrast to higher-dimensional faces of the SAC$_{\N}$), which means at least for extreme rays the approximation of SSA by KC is exact. This led \cite{He:2022bmi} to speculate that this equivalence might hold for arbitrary $\N$. However, this is not true, and we will construct the simplest counterexamples for $\N=6$.

\begin{table}[]
    \centering
    \small
    \begin{tabular}{|l|l|}
    \hline
      $\ersa$   & {\footnotesize All the extreme rays (ERs) of the SAC}  \\
      \hline
      $\erkc$   & {\footnotesize ERs of the SAC that satisfy KC} \\
      \hline
      $\erssa$   & {\footnotesize ERs of the SAC that satisfy SSA} \\
      \hline
      $\erq$   & {\footnotesize ERs of the SAC realizable by arbitrary quantum states} \\
      \hline
      $\ersta$   & {\footnotesize ERs of the SAC realizable by stabilizer states} \\
      \hline
      $\erhg$   & {\footnotesize ERs of the SAC realizable by hypergraph models} \\
      \hline
      $\erh$   & {\footnotesize ERs of the SAC realizable by graph models}  \\
      \hline
      $\ert$   & {\footnotesize ERs of the SAC realizable by tree graph models} \\
      \hline
      $\erst$   & {\footnotesize ERs of the SAC realizable by simple tree graph models} \\
      \hline
      $\erstar$  & {\footnotesize ERs of the SAC realizable by star graph models} \\
      \hline
    \end{tabular}
    \caption{Summary of the definitions of various classes of extreme rays of the SAC for a fixed value of $\N$ (to keep the table cleaner, we dropped the explicit dependence on $\N$ of each set).}
    \label{tab:R_def}
\end{table}

\begin{table}[]
    \centering
    \small
    \begin{tabular}{|l|l|l|}
      \hline
      $\erstar\, \subseteq\, \erst$   & {\footnotesize Any star graph   model is a simple tree graph model} & \cite{Bao:2015bfa,Hernandez-Cuenca:2022pst} \\
      \hline
      $\erst\, \subseteq\, \ert$   & {\footnotesize Any simple tree graph model is a tree graph model} & \cite{Bao:2015bfa,Hernandez-Cuenca:2022pst} \\
      \hline
      $\ert\, \subseteq\, \erh$   & {\footnotesize Any tree graph model is a graph model}  & \cite{Bao:2015bfa,Hernandez-Cuenca:2022pst} \\
      \hline
      $\erh\, \subseteq\, \erhg$   & {\footnotesize Any graph model is a hypergraph model} & \cite{Bao:2020zgx} \\
      \hline
      $\erhg\, \subseteq\, \ersta$   & {\footnotesize Any hypergraph model can be realized by a stabilizer state} & \cite{Walter:2020zvt} \\
      \hline
      $\ersta\, \subseteq\, \erq$   & {\footnotesize Any stabilizer state is a quantum state} & \cite{Linden:2013kal} \\
      \hline
      $\erq\, \subseteq\, \erssa$   & {\footnotesize Any quantum state satisfies SSA} & \cite{Lieb:1973cp,1193790} \\
      \hline
      $\erssa\, \subseteq\, \erkc$   & {\footnotesize KC is implied by SSA} & \cite{Nielsen_Chuang_2010,He:2022bmi}\\
      \hline
      $\erkc\, \subseteq\, \ersa$   & {\footnotesize KC is a condition on the set of SA instances saturated by an ER}  & \cite{Nielsen_Chuang_2010,He:2022bmi} \\
      \hline
    \end{tabular}
    \caption{Inclusion relations for the various classes of extreme rays of the SAC. For each inclusion, we added references to the papers where it was proven, or where the structure in entropy space which defines the class was first introduced (when the inclusion follows automatically from the definition of the class).}
    \label{tab:R_inclusions}
\end{table}

\paragraph{Preview of results:}
The set $\erssa^6$ is much richer than its counterparts for lower parties, and we can use our results in this paper to investigate various conjectures and questions that have appeared in previous work. In \Cref{tab:R_def} we review the definition of various useful subsets of $\ersa$, depending on whether they are realizable or, if they are, the class of quantum states that realizes them. These subsets satisfy the following chain of inclusions (see \Cref{tab:R_inclusions} for more details):
\begin{equation}
    \erstar\, \subseteq\, \erst\, \subseteq\, \ert\, \subseteq\, \erh\, \subseteq\, \erhg\, \subseteq\, \ersta\, \subseteq\, \erq\, \subseteq\, \erssa\, \subseteq\, \erkc\, \subseteq\, \ersa .
\end{equation}
For $\N\leq 5$, \Cref{tab:R_strict_inclusions_equalities} summarizes which inclusions are strict and which inclusions are in fact equalities. For $\N\geq 6$, the current knowledge can instead be summarized as follows:
\begin{equation}
\label{eq:N6_R_nesting}
    \erstar\, \subset\, \erst\, \subset\, \ert\, \subseteq\, \erh\, \subset\, \erhg\, \subseteq\, \ersta\, \subseteq\, \erq\, \subseteq\, \erssa\, \subset\, \erkc\, \subset\, \ersa ,
\end{equation}
where $\erh\, \subset\, \erhg$ was shown in \cite{He:2023cco}, and $\erssa\, \subset\, \erkc$ will be shown in this paper.

\begin{table}[]
    \centering
    \small
    \begin{tabular}{|l|l|l|}
    \hline
    $\N=2$ & $\erstar =\, \ersa$ & \cite{Bao:2015bfa,Hernandez-Cuenca:2022pst} \\
    \hline
    $\N=3$ & $\erstar =\, \erkc \subset\, \ersa$ & \cite{Bao:2015bfa,Hernandez-Cuenca:2022pst} \\    
    \hline
    $\N=4$ & $\erstar =\, \erkc \subset\, \ersa$ & \cite{Bao:2015bfa,Hernandez-Cuenca:2022pst} \\
    \hline
    $\N=5$ & $\erstar \subset\, \erst  =\, \erkc \subset\, \ersa$ & \cite{Cuenca:2019uzx,Hernandez-Cuenca:2022pst} \\
    \hline
    $\N=6$   & {\footnotesize See \eqref{eq:N6_R_nesting} and related references} & \\
    \hline
    \end{tabular}
    \caption{Summary of the essential equivalencies and strict inclusions among the various classes for different values of $\N$. Further equivalences follow automatically from the inclusions in \Cref{tab:R_inclusions}. 
    }
    \label{tab:R_strict_inclusions_equalities}
\end{table}

Furthermore, a conjecture first proposed in \cite{Hernandez-Cuenca:2022pst} implies that $\ert = \erh$, which we are able to explore in the $\N=6$ case in this paper. As we will see in more detail, the elements of $\erssa^6$ can be organized into 208 orbits under the symmetries of the problem, and out of these, 52 violate at least one known \textit{holographic entropy inequality} (HEI) and are therefore not in $\erh^6$. Of the remaining 156 orbits, we managed to construct an explicit graph model for 150 orbits, thereby proving that they are in $\erh^6$. While 148 of these graph models indeed have a tree topology, two of them do not, therefore just falling short of verifying the conjecture $\ert = \erh$ for $\N=6$. Of course, this result does not disprove the conjecture either, since we did not prove that there cannot be tree graph models realizing those two entropy vectors. Moreover, to get a definitive answer at $\N=6$, one would also need to determine whether the remaining six ``mystery'' orbits are holographic or not, and if so, whether they can be realized by trees. 

Finally, we stress that our algorithm relies only on the existence of a partial order among the inequalities. This means it is not in any way specific to entropy inequalities such as SA. In fact, it can be used in any context where one is interested in finding not all extreme rays, but rather only the subset of extreme rays which correspond to down-sets in the poset. 

The structure of the paper is as follows. In \S\ref{sec:mip}, we review the necessary definitions about the mutual information poset, how it can be used to formulate KC, and a few additional results from our previous work that we will use to simplify the computation of $\erssa^6$. The algorithm that we use is then presented in \S\ref{sec:algorithm}. We stress that this section is written independently from the set-up of the SAC reviewed in \S\ref{sec:mip}, and thus our algorithm can be immediately utilized for other applications. To exemplify how the algorithm works, we then re-compute the set $\erssa^5$ in \S\ref{sec:five}. The result of the computation for $\N=6$ is presented in \S\ref{sec:results}, along with a preliminary  analysis of this data and all the holographic graph models that we were able to construct.\footnote{\,Further details regarding the various attributes of $\erssa^6$ are also presented in \cite{rays-notes}, and the actual graph constructions in \cite{graph-constructions}.} Readers who are only interested in the results of the $\N=6$ case, and not in the details of the algorithm, can skip earlier sections and proceed directly to \S\ref{sec:results}. 
Finally, in \S\ref{sec:discussion} we comment on open questions and future directions of investigation.

\section{Preliminaries}
\label{sec:mip}

An instance of SA (see \eqref{eq:sa}) can be written as 
\begin{align}
    \mi(\uJ:\uK) \equiv \ent_{\uJ} + \ent_{\uK} - \ent_{\uJ\uK}\geq 0,  
\end{align}
where $\mi(\uJ:\uK)$ is the \textit{mutual information} (MI) between subsystems $\uJ$ and $\uK$. Clearly, an instance of SA is saturated if and only if the corresponding MI instance vanishes. For any given $\N$, we denote by $\mgs$ the set of all MI instances.

We will review in this section the definition of Klein's condition for a subset $\x\subseteq\mgs$, and how this can be formulated in terms of a poset. This will make manifest that in the search for $\erkc^{^{\N}}$ we can restrict the search space from the power-set $2^{\mgs}$ to the set of down-sets $\x\subseteq\mgs$ which include all minimal elements of the poset without containing any maximal, or even next-to-maximal, elements. The presentation will follow \cite{He:2022bmi} closely, and we refer the reader to that paper for more details.

\subsection{The mutual information poset and Klein's condition}
\label{subsec:mi-poset-and-kc}

Given a face $\face$ of the SAC$_{\N}$, consider an entropy vector $\vec{\ent}\in\text{int}(\face)$ and the set of all MI instances that vanish for $\vec{\ent}$, i.e., 
\begin{equation}
\label{eq:vanishing_mi}
    \van(\vec{\ent})= \{\mi(\uJ:\uK)\in\mgs,\; \ \mi(\uJ:\uK)(\vec{\ent})=0\}.
\end{equation}
Notice that this set does not depend on the specific choice of $\vec{\ent}$ in the interior of $\face$, and therefore we can naturally associate a set $\van(\face)$ to a face $\face$. Furthermore, we stress that here $\face$ and $\vec{\ent}$ are not assumed to be realizable, and we can extend the map \eqref{eq:vanishing_mi} to the full entropy space.

Since in quantum mechanics an MI instance vanishes if and only if the density matrix factorizes, we have the chain of implications
\begin{equation}
    \mi(\uI:\uJ\uK)=0\; \implies \; \rho_{\uI\uJ\uK}=\rho_{\uI}\otimes\rho_{\uJ\uK} \; \implies \; \rho_{\uI\uK}=\rho_{\uI}\otimes\rho_{\uK} \; \implies \; \mi(\uI:\uK)=0 .
\end{equation}
Motivated by this simple observation, we then define the following necessary condition for the realizability of a face of the SAC$_{\N}$.
\begin{defi}[Klein's condition (KC)]\label{def:KC}
    A face $\face$ of the \emph{SAC}$_{\N}$ satisfies Klein's condition if for any three disjoint subsets $\uI,\uJ$, and $\uK$, 
    \begin{align}
    \label{eq:kc_def}
        \mi(\uI:\uJ\uK) \in \van(\face)\quad\implies\quad \mi(\uI:\uK) \in \van(\face).
    \end{align} 
    If a face $\face$ satisfies \emph{KC} we call it a \emph{KC}-face, and in the 1-dimensional case, a \emph{KC} extreme ray \emph{(KC-ER)}.
\end{defi}
\noindent Notice that, given SA, this condition immediately follows from SSA, which can be written as
\begin{align}
    \mi(\uI:\uJ\uK) \geq \mi(\uI:\uK).
\end{align}
We can then interpret KC as an approximation of SSA, in the sense that any SSA-compatible face of the SAC$_{\N}$ is necessarily a KC-face. As anticipated in the Introduction, the advantage of this approximation is that it characterizes a subset of faces of the SAC$_{\N}$ that \textit{could} be SSA-compatible purely in terms of SA.

Having defined KC, we now review how it can be conveniently rephrased in terms of a partially ordered structure. Recall that a poset $\p=(X,\preceq)$ is defined as a set $X$ with a binary relation $\preceq$ that is reflexive, antisymmetric, and transitive.\footnote{\,That is, for all $x,y,z\in X$ we have (i) $x\preceq x$; (ii) if $x\preceq y$ and $y\preceq x$, then $x=y$; and finally (iii) if $x\preceq y$ and $y\preceq z$, then $x\preceq z$.} We define the \textit{mutual information poset} (MI-poset) as follows.

\begin{defi}[Mutual information poset (MI-poset)]
    For any given $\N$, the mutual information poset $(\mgs,\preceq)$ is the set of all \emph{MI} instances $\mgs$ with partial order given by the relation
    \begin{align}
    \label{eq:mi_order}
        \mi(\uI:\uK) \preceq \mi(\uI':\uK') \quad\iff\quad \uI \subseteq \uI' \;\;\text{and}\;\; \uK \subseteq \uK' \quad\text{or}\quad \uI \subseteq \uK' \;\;\text{and}\;\; \uK \subseteq \uI'.
    \end{align}
\end{defi}

\noindent As a simple example, we have depicted the MI-poset for $\N=2$ in \Cref{fig:MIposet2} using a Hasse diagram. The lines denote covering relations, where an element $x$ is joined to an element $y$ if $x \prec y$, and there does not exist $z$ such that $x \prec z \prec y$. 

Given a poset $\p$, a \textit{down-set} is defined as a subset of $Y\subseteq X$ such that for any $ y\in Y$ ,
\begin{equation}
   z\preceq y \; \implies \;  z\in Y .
\end{equation}
It then follows from the definition of the partial order in the MI-poset \eqref{eq:mi_order} and the definition of KC given in \eqref{eq:kc_def} that a face $\face$ of the SAC$_{\N}$ is a KC-face only if $\van(\face)$ is a down-set. It is important to notice that given an arbitrary down-set $\ds$ of the MI-poset, it is \textit{not} necessarily the case that there exists a face $\face$ such that $\van(\face)=\ds$. An analogous statement also holds for arbitrary subsets $\x\subseteq\mgs$, and the underlying reason is in fact the same: Such arbitrary $\ds$ and $\x$ do not necessarily respect the dependences among SA instances. For now, however, it is sufficient to notice that by focusing on down-sets of the MI-poset, we can restrict the search space for the brute force computation of $\erssa^{^{\N}}$ from the power-set $2^\mgs$ to the lattice of down-sets\footnote{\,The set of all down-sets of any poset, with partial order given by inclusion, has the structure of a lattice, where meet and join correspond respectively to intersection and union \cite{birkhoff1967lattice}.} of the MI-poset, which is already significantly smaller. For example, for $\N=3$ we have $|2^{\mgs}|=2^{25}=33,\!554,\!432$, while the down-set lattice has cardinality $11,\!422$.

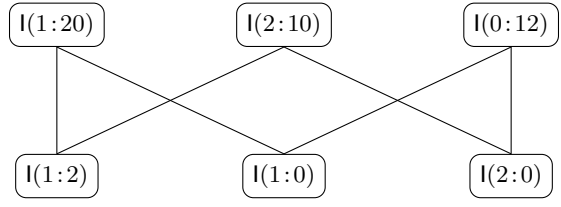
\begin{figure}[tb]   
\begin{center}
    \begin{tikzpicture}
    \node[draw, rounded corners] (ao) at (0,0) {{\footnotesize $\mi(1\!:\!0)$}};
    \node[draw, rounded corners] (abo) at (-3,2) {{\footnotesize $\mi(1\!:\!20)$}};
    \node[draw, rounded corners] (ab) at (-3,0) {{\footnotesize $\mi(1\!:\!2)$}};
     \node[draw, rounded corners] (bao) at (0,2) {{\footnotesize $\mi(2\!:\!10)$}};
    \node[draw, rounded corners] (bo) at (3,0) {{\footnotesize $\mi(2\!:\!0)$}};
    \node[draw, rounded corners] (oab) at (3,2) {{\footnotesize $\mi(0\!:\!12)$}};
    \draw[-] (bao.south) -- (ab.north);
    \draw[-] (abo.south) -- (ao.north);
    \draw[-] (bao.south) -- (bo.north);
    \draw[-] (oab.south) -- (ao.north);
    \draw[-] (oab.south) -- (bo.north);
    \draw[-] (abo.south) -- (ab.north);
    \end{tikzpicture}
\end{center}   
    \caption{The Hasse diagram of the MI-poset for $\N=2$.}
    \label{fig:MIposet2}
\end{figure}

\subsection{The Bell pairs theorem and additional constraints}
\label{subsec:constraints}

Having explained how to reduce the search space for $\erssa^{^{\N}}$ to down-sets of the MI-poset, we now review two constraints that can be used to ignore certain families of down-sets. The first result was already anticipated in the Introduction and allows us to focus on a specific face $\face_*$ of the SAC$_{\N}$ which includes all KC-ERs that are not realized by Bell pairs. Since this face is realizable in quantum mechanics, it is a KC-face, and the set $\van(\face_*)$ must be a down-set. In particular, 
$\van(\face_*)$ is the set of minimal elements of the MI-poset \cite{He:2022bmi}.

\begin{thm}[Bell pairs theorem]\label{thm:bell-pair}
    For any given number of parties $\N$, any \emph{KC}-compatible extreme ray of the \emph{SAC}$_\N$ which is not realized by a single Bell pair obeys the condition
    \begin{align}
    \label{eq:er_face}
        \mi(\ell:\ell') = 0 \quad\text{for all $\ell,\ell' \in \psys{\N}$.}
    \end{align}
\end{thm}
\begin{proof}
    See Corollary 1 in \cite{He:2022bmi}.
\end{proof}

This theorem implies that to find extreme rays which do not correspond to Bell pairs, we should only consider down-sets of the MI-poset which include all minimal elements, i.e., the MI instances between single parties. However, notice that since $\face_*$ contains \textit{all} KC-ERs except for those realized by Bell pairs, in particular it also contains trivial lifts of KC-ERs for any $\N'$ such that $3\leq\N'\leq\N$ \cite{He:2022bmi}. 

Suppose now that we want to compute $\erssa^{^{\N}}$ for some $\N$, and that the sets $\erssa^{^{\N'}}$ for all $\N'\leq\N$ are already known. The second result that we review will allow us to only consider down-sets which are guaranteed to correspond to ``genuine $\N$-party'' KC-ERs, therefore avoiding the reconstruction of extreme rays that we already know. Notice that the maximal elements of the MI-poset take the form
\begin{align}
    \mi(\uJ : \uJ^c) = \ent_{\uJ} + \ent_{\uJ^c} -  \ent_{\psys{\N}} = 2 \ent_{\uJ}, 
\end{align}
where $\uJ^c \equiv \psys{\N} \setminus \uJ$, and we utilized the assumption that the system on all $\psys{\N}$ parties is pure.\footnote{\,This implies that $\ent_{\psys{\N}}= \ent_{\varnothing} = 0$ and $\ent_{\uJ^c}=\ent_{\uJ}$.} Consider now a face $\face$ of the SAC$_{\N}$ such that $\van(\face)$ includes one such element. Assuming for simplicity that $\uJ$ does not include the purifier $\ell=0$ (if it does we can simply repeat the same argument for $\uJ^c$), we can then denote $\uJ$ by just $\J$ and have $\ent_{\J}=0$ for any vector $\vec{\ent}\in\face$. If $\vec{\ent}$ is realizable by some density matrix $\rho_{\N}$, then we can write it as
\begin{equation} \label{eq:prod-state}
\rho_{\N}= \ket{\psi}\bra{\psi}_{\J} \otimes\rho_{\J^c},
\end{equation}
where $\ket{\psi}_{\J}$ is a pure state and $\J^c \equiv [\N]\setminus\J$ (by itself purified by the purifier $0$). The density matrix $\rho_{\N}$ is then a state of a system composed of two uncorrelated subsystems on $|\J|$ and $|\J^c|$ parties respectively. 
We can rewrite the entropy vector of this density matrix to be 
\begin{align}
    \vec\ent(\rho_\N) = \vec\ent\big( \ket{\psi}\bra{\psi}_{\J} \otimes\rho_{\J^c} \big) = \vec\ent \big( \ket{\psi}\bra{\psi}_{\J} \otimes |0 \rangle \langle 0|_{\J^c} \big) + \vec{\ent}\big( |0 \rangle \langle 0|_{\J} \otimes \rho_{\J^c} \big) ,
\end{align}
where $\ket{0}\bra{0}_{\J}$ is a ``completely decorrelated'' state for subsystem $\J$.\footnote{\,In other words, $\ket{0}_\J$ is a tensor product of states for the individual parties in $\J$.} If $\face$ is an extreme ray, any entropy vector $\vec\ent(\rho_\N)$ generating it cannot be the sum of two distinct nontrivial vectors within the SAC. Therefore, the only possibility is that the first term is the null vector, which implies
\begin{equation}
    \ket{\psi}_\J = \ket{0}_{\J} \quad\implies\quad \rho_{\N}=|0\rangle\langle 0|_{\J}\otimes\rho_{\J^c} .
\end{equation}
Thus, the extreme ray associated to $\rho_\N$ is a ``canonical'' lift (see \cite{He:2022bmi}) to $\N$ parties of the extreme ray associated to $\rho_{\J^c}$ of the SAC for $|\J^c|$ parties.

Indeed, we can use a similar argument to see that we can even ignore any down-set involving next-to-maximal elements in the MI-poset. Suppose our down-set includes an element of the form
\begin{align}
    \mi(\J:\J^c) = \ent_\J + \ent_{\J^c} - S_{[\N]},
\end{align}
and we have assumed without loss of generality that the party not appearing in the arguments is the purifier. As in our argument above, assuming $\vec\ent(\rho_\N)$ is realizable by some quantum state $\rho_{\N}$, this means we can write
\begin{align}
\label{eq:rho-factorization}
    \rho_{\N} = \rho_\J \otimes \rho_{\J^c}.
\end{align}
The entropy vector associated to this density matrix can be rewritten as
\begin{align}\label{eq:next-to-max}
     \vec\ent(\rho_\N) = \vec\ent\big( \rho_{\J} \otimes\rho_{\J^c} \big) = \vec\ent \big( \rho_{\J} \otimes |0 \rangle \langle 0|_{\J^c} \big) + \vec{\ent}\big( |0 \rangle \langle 0|_{\J} \otimes \rho_{\J^c} \big) .
\end{align}
This implies that $\vec\ent(\rho_\N)$ is the sum of two nontrivial entropy vectors in the $\N$-party entropy space, and hence $\vec\ent(\rho_\N)$ cannot be an extreme ray. Notice that we cannot use this argument to further exclude MI-poset elements $\mi(\J:\K)$ below the next-to-maximal ones, since even if we still have a factorization of the density matrix as in \eqref{eq:rho-factorization}, it would not involve all the $\N$ parties. Finally, notice that for these arguments, we used the assumption that the face $\face$ is realizable. The result however can be extended to any KC-face, independently from realizability, and we refer to reader to Section 3.4 of \cite{He:2022bmi} for the proof and a more detailed discussion. 

As an explicit example of these two results consider again the case of $\N=3$. The number of down-sets that include all minimal elements is $8,\!695$, and if we further demand that the down-sets do not include any maximal elements we are left with $2^{12}=4,\!096$ possibilities. These correspond to all possible subsets of the antichain
\begin{equation}
\label{eq:n3_antichain}
    \mathcal{A}=\{\mi(\ell:\ell'\ell''),\; \text{for all distinct}\; \ell,\ell',\ell''\in\psys{3}\}.
\end{equation}
However, consider for example a down-set $\ds$ composed of all minimal elements and $\mi(1:23)$. Notice that 
\begin{equation}
\label{eq:n3_linear_dep}
    \mi(1:23)+\mi(1:0)=\mi(1:230)=2\ent_{1}.
\end{equation}
Since $\mi(1:0)$ is a minimal element, both terms on the l.h.s.\ of \eqref{eq:n3_linear_dep} are included in $\ds$ while the term on the r.h.s.\ is not, so \eqref{eq:n3_linear_dep} implies that there is no face $\face$ of SAC$_3$ with $\van(\face)=\ds$, and that for any face such that $\van(\face)=\ds'$ with $\ds'\supset\ds$, we must have $\mi(1:230)\in\ds'$. Since $\mi(1:230)$ is a maximal element, we are then forced to consider only down-sets that do not contain $\mi(1:23)$. Permuting the parties in \eqref{eq:n3_linear_dep} to obtain analogous relations for the other elements of $\mathcal{A}$ in \eqref{eq:n3_antichain}, this implies that the only case that we need to consider is the one where $\ds$ is precisely the set of minimal elements. In fact, from \eqref{eq:special_face_dim} one can see that for $\N=3$, $\text{dim}(\face_*)=1$, and for this choice of $\ds$ we get the ray generated by $\vec{\ent}=(1,1,1,2,2,2,1)$, which is an extreme ray and is the \textit{only} KC-ER of SAC$_3$. Therefore, in the simple case of $\N=3$, the interplay of the constraints discussed in this section with the dependences among the MI instances reduces the size of the search space for the brute force strategy from $33,\!554,\!432$ to $1$! The combination of these constraints will indeed be a key aspect of the algorithm described in the next section. 

Finally, let us conclude with another comment regarding the maximal elements of the MI-poset. For arbitrary $\N$, we can generalize \eqref{eq:n3_linear_dep} to be
\begin{equation}
\label{eq:positive_comb}
    \mi(\uJ:\uK)+\mi(\uJ:(\uJ\uK)^c)=\mi(\uJ:\uJ^c) ,
\end{equation}
which shows that an arbitrary maximal element (on the r.h.s.) is a \textit{positive} linear combination of other MI instances. Translating each MI in \eqref{eq:positive_comb} to the corresponding SA, this shows that each SA instance corresponding to a maximal element of the poset is a \textit{redundant inequality}, i.e., it is not a facet of the SAC$_{\N}$. In other words, we could simply ignore all these SA instances and the cone would be exactly the same. However, for the sake of convenience in our algorithm, which we will explain below, we will instead choose to keep such redundant inequalities. Geometrically, we can interpret the set of redundant inequalities as specifying a ``region'' of the cone whose extreme rays we want to ignore. This will be another key element of the algorithm that we will now present.

\section{The algorithm}
\label{sec:algorithm}

In this section, we present the algorithm that we will use to find $\erkc^6$. Since nothing about this algorithm is specific to the set-up of the subadditivity cone described in the previous section, we will present it more generally, such that it can be applied to other contexts without any adaptation. In general terms, our algorithm can be used when one has a polyhedral cone specified by a set of linear inequalities with partial order among them, and one wants to find only a particular subset of extreme rays, characterized by the fact that for each extreme ray the set of inequalities that are saturated is a down-set in the poset.\footnote{\,To make the algorithm more robust, we will make a few additional assumptions, which will be presented in detail below.}

We anticipate that this algorithm is typically most useful in situations where the poset has ``large height and small width.'' In the limiting situation where the poset is a chain, our algorithm is very efficient (although somewhat trivial). In the opposite situation, where the poset is an antichain, it is instead very inefficient since it reduces to the brute force strategy mentioned in the Introduction, and one should instead use one of the standard algorithms for ``polyhedral representation conversion'' (although in general even these algorithms are not efficient) \cite{survey1980}.\footnote{\,Notice however that this inefficiency is related to the fact that the subset of extreme rays found by the algorithm becomes the full set of all extreme rays of the cone.} In intermediate situations, the efficiency will depend on the details of the problem at hand. Nevertheless, even for the specific problem of finding $\erkc$ indicated in the preceding section, where the height (approximately $\N$) is exponentially smaller by a factor of roughly $2^D/ \sqrt{N}$ than the width (approximately $\binom{\N}{\N / 2}$), the algorithm remains remarkably useful. In any case, a feature of our algorithm is the flexibility of how it can be used in practice, and the possibility of easily combining it with another conversion algorithm. Essentially, one first runs our algorithm to extract a collection of lower-dimensional faces of the cone that is guaranteed to contain all the ``down-set extreme rays.'' Then one can find all the extreme rays for each of these faces using standard techniques and finally check which ones correspond to down-sets (which is a straightforward computation). Finally, we stress that even in this situation, the algorithm will provide useful information about ``excluded regions'' on the individual faces that in many cases can be used to simplify the derivation of the down-set extreme rays using standard conversion algorithms.

\subsection{Set-up}
\label{subsec:set-up}

\begin{figure}[tbp]
    \centering
    \begin{subfigure}{0.49\textwidth}
    \centering
    \begin{tikzpicture}[scale=0.5]
    \fill[fill=blue!15!] (0,0) -- (8,2) -- (3,6) ;
    \fill[fill=red!15!] (0,0) -- (3,6) -- (-1,7) ;
    \draw[->, gray, thick] (0,0) -- (3,6) ;
    \fill[fill=yellow!80, opacity=0.5] (0,0) -- (-1,7) -- (4,3) ;
    \fill[fill=green!20, opacity=0.5] (0,0) -- (8,2) -- (4,3) ;
    \draw[->, gray, thick] (0,0) -- (4,3) ;
    \draw[->, gray, thick] (0,0) -- (-1,7) ;
    \draw[->, gray, thick] (0,0) -- (8,2) ;
    \node (f1) at (1,4.5) {{\footnotesize $\face_{_{\!1}}$}}; 
    \node (f2) at (1.6,5.6) {{\footnotesize $\face_{_{\!2}}$}}; 
    \node (f3) at (4.5,4) {{\footnotesize $\face_{_{\!3}}$}}; 
    \node (f4) at (5,2) {{\footnotesize $\face_{_{\!4}}$}};
    \end{tikzpicture}
    \subcaption[]{}
    \end{subfigure}
    \begin{subfigure}{0.49\textwidth}
    \centering
    \begin{tikzpicture}[scale=0.5]
    \node[draw, rounded corners] (ao) at (0,0) {{\footnotesize $1$}};
    \node[draw, rounded corners] (abo) at (-3,3) {{\footnotesize $2$}};
    \node[draw, rounded corners] (ab) at (-3,0) {{\footnotesize $3$}};
    \node[draw, rounded corners] (bao) at (3,0) {{\footnotesize $4$}};
    \draw[-] (abo.south) -- (ab.north);
    \end{tikzpicture}
    \subcaption[]{}
    \end{subfigure}
    \par\bigskip
    \begin{subfigure}{0.49\textwidth}
    \centering
    \begin{tikzpicture}[scale=0.5]
    \node[draw, align=center, rounded corners] (point) at (0,2) {{\footnotesize $\{1,2,3,4\}$}};
    \node[draw, align=center, rounded corners] (line1) at (-6,0) {{\footnotesize $\{1,2\}$}};
    \node[draw, align=center, rounded corners] (line2) at (-2,0) {{\footnotesize $\{2,3\}$}};
    \node[draw, align=center, rounded corners] (line3) at (2,0) {{\footnotesize $\{3,4\}$}};
    \node[draw, align=center, rounded corners] (line4) at (6,0) {{\footnotesize $\{1,4\}$}};
    \node[draw, align=center, rounded corners, fill=red!15!] (face1) at (-6,-2) {{\footnotesize $ \{2\}$}};
    \node[draw, align=center, rounded corners, fill=blue!15!] (face2) at (-2,-2) {{\footnotesize $ \{3\}$}};
    \node[draw, align=center, rounded corners, fill=green!15!] (face3) at (2,-2) {{\footnotesize $ \{4\}$}};
    \node[draw, align=center, rounded corners, fill=yellow!30!] (face4) at (6,-2) {{\footnotesize $ \{1\}$}};
    \node[draw, align=center, rounded corners] (cone) at (0,-4) {{\footnotesize $\varnothing$}};
    \draw[-] (point.south) -- (line1.north);
    \draw[-] (point.south) -- (line2.north);
    \draw[-] (point.south) -- (line3.north);
    \draw[-] (point.south) -- (line4.north);
    \draw[-] (line1.south) -- (face1.north);
    \draw[-] (line1.south) -- (face4.north);
    \draw[-] (line2.south) -- (face1.north);
    \draw[-] (line2.south) -- (face2.north);
    \draw[-] (line3.south) -- (face2.north);
    \draw[-] (line3.south) -- (face3.north);
    \draw[-] (line4.south) -- (face3.north);
    \draw[-] (line4.south) -- (face4.north);
    \draw[-] (face1.south) -- (cone.north);
    \draw[-] (face2.south) -- (cone.north);
    \draw[-] (face3.south) -- (cone.north);
    \draw[-] (face4.south) -- (cone.north);
    \end{tikzpicture}
    \subcaption[]{}
    \end{subfigure}
    \begin{subfigure}{0.49\textwidth}
    \centering
    \begin{tikzpicture}[scale=0.5]
    \node[draw, align=center, rounded corners] (point) at (0,4) {{\footnotesize $\{1,2,3,4\}$}};
    \node[draw, align=center, rounded corners] (ds123) at (-4,2) {{\footnotesize $\{1,2,3\}$}};
    \node[draw, align=center, rounded corners] (ds134) at (4,2) {{\footnotesize $\{1,3,4\}$}};
    \node[draw, align=center, rounded corners] (ds234) at (0,2) {{\footnotesize $\{2,3,4\}$}};
    \node[draw, align=center, rounded corners] (ds23) at (-6,0) {{\footnotesize $\{2,3\}$}};
    \node[draw, align=center, rounded corners] (ds13) at (-2,0) {{\footnotesize $\{1,3\}$}};
    \node[draw, align=center, rounded corners] (ds34) at (2,0) {{\footnotesize $\{3,4\}$}};
    \node[draw, align=center, rounded corners] (ds14) at (6,0) {{\footnotesize $\{1,4\}$}};
    \node[draw, align=center, rounded corners] (ds3) at (-4,-2) {{\footnotesize $ \{3\}$}};
    \node[draw, align=center, rounded corners] (ds1) at (0,-2) {{\footnotesize $ \{1\}$}};
    \node[draw, align=center, rounded corners] (ds4) at (4,-2) {{\footnotesize $ \{4\}$}};
    \node[draw, align=center, rounded corners] (cone) at (0,-4) {{\footnotesize $\varnothing$}};
    \draw[-] (point.south) -- (ds123.north);
    \draw[-] (point.south) -- (ds134.north);
    \draw[-] (point.south) -- (ds234.north);
    \draw[-] (ds123.south) -- (ds23.north);
    \draw[-] (ds123.south) -- (ds13.north);
    \draw[-] (ds134.south) -- (ds13.north);
    \draw[-] (ds134.south) -- (ds14.north);
    \draw[-] (ds134.south) -- (ds34.north);
    \draw[-] (ds234.south) -- (ds23.north);
    \draw[-] (ds234.south) -- (ds34.north);
    \draw[-] (ds23.south) -- (ds3.north);
    \draw[-] (ds13.south) -- (ds1.north);
    \draw[-] (ds13.south) -- (ds3.north);
    \draw[-] (ds14.south) -- (ds1.north);
    \draw[-] (ds14.south) -- (ds4.north);
    \draw[-] (ds34.south) -- (ds3.north);
    \draw[-] (ds34.south) -- (ds4.north);
    \draw[-] (ds1.south) -- (cone.north);
    \draw[-] (ds3.south) -- (cone.north);
    \draw[-] (ds4.south) -- (cone.north);
    \end{tikzpicture}
    \subcaption[]{}
    \end{subfigure}
    \caption{A 3-dimensional pointed cone $\cone$ in $\mathbb{R}^3$ specified by (a) a set of four (non-redundant) inequalities, (b) a choice of a poset $\p$ for the inequalities, (c) the lattice of faces of $\cone$, and (d) the lattice of down-sets of $\p$. The combinatorial automorphism groups are $G_{\cone}\simeq D_4$ and $G_{\p}\simeq \mathbb Z_2$. Notice that $G_{\p}=\{e,(14)\}$ and is not a subgroup of $G_{\cone}$, implying that $G_{\cone}\wedge G_{\p}=\{e\}$. The D-ERs are $\{2,3\}, \{3,4\}$ and $\{1,4\}$, whereas $\{1,2\}$ is an ER but not a D-ER. On the other hand, the down-sets $\{3,4\}$ and $\{1,3\}$ belong to the same $G_{\p}$-orbit, but $\{1,3\}$ is not an ER.}
    \label{fig:3d_cone}
\end{figure}
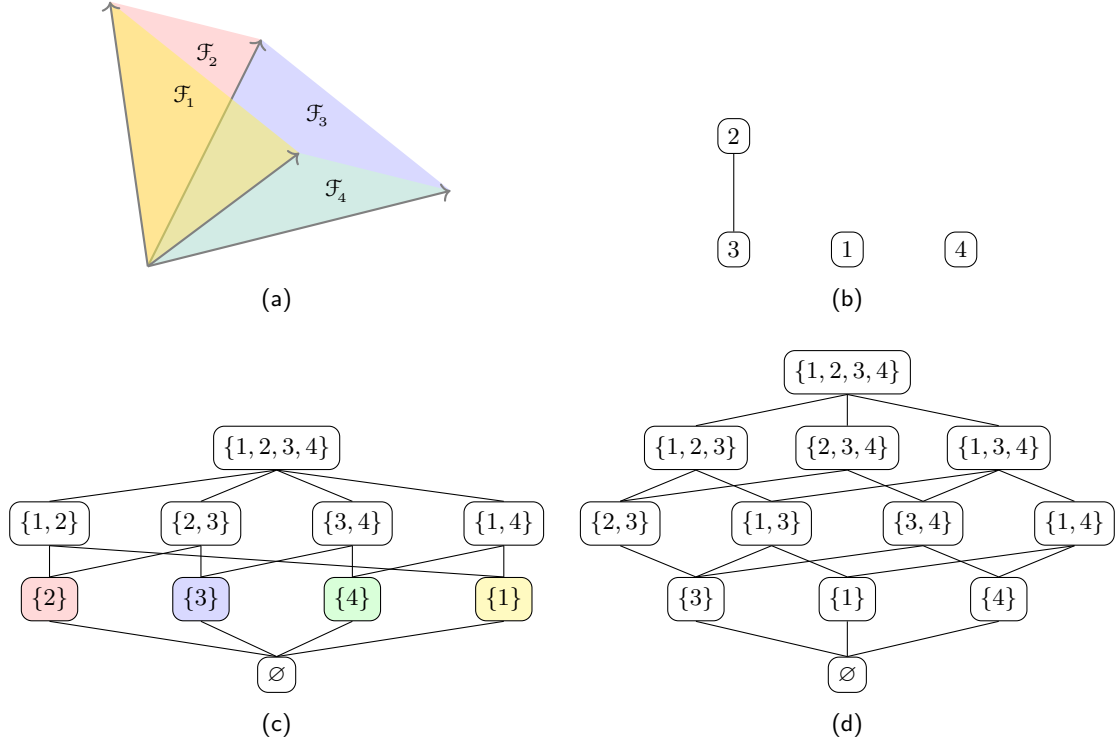

We begin by describing the set-up for our algorithm and by reviewing a few definitions and basic facts about polyhedral cones and systems of linear inequalities. For more details we refer the reader to standard references on the subject (e.g., see \cite{z-lop-95}). As a concrete and simple example of the set-up and the concepts introduced, we refer the reader to \Cref{fig:3d_cone}.

\paragraph{Polyhedral cone:} 

Consider a vector space $\mathbb{R}^n$, along with a finite collection of linear inequalities that specify a full-dimensional pointed cone $\cone$.\footnote{\,Any cone specified by a finite set of inequalities is necessarily convex and polyhedral. If the cone is not full-dimensional we imagine to first ``dimensionally reduce'' the set-up to the subspace spanned by the cone. If it is not pointed, it is the Minkowski sum of a lower-dimensional cone and a linear subspace, and we can focus on this lower-dimensional cone.}  We will write an inequality as $\E_i(\vec{v}) \geq 0$, where $\E_i$ is an element of the dual space of $\mathbb{R}^n$ and $\vec{v}$ denotes an arbitrary vector in $\mathbb{R}^n$. Let $\mgs$ denote the set of dual vectors $\E_i$, with $i \in [k]$ and $k = |\mgs|$. A \textit{face} $\face$ of $\cone$ is defined as the intersection of $\cone$ with a hyperplane $\hyp$ in $\mathbb{R}^n$ such that $\cone$ is entirely contained in one of the two (closed) half-spaces specified by $\hyp$.\footnote{\,The full cone is also a face, obtained in the degenerate case where $\E_i=\vec{0}$. In what follows we will ignore this face and always assume $\E_i\neq\vec{0}$.} A codimension-1 face is called a \textit{facet}, and a $1$-dimensional face an \textit{extreme ray}.

In general, given an arbitrary collection of inequalities that specifies a full-dimensional pointed cone, it is not always the case that each inequality determines a facet. An inequality $\E_i(\vec{v})\geq 0$ such that $\face_i$ is not a facet is called a \textit{redundant inequality}. Geometrically, a redundant inequality $\E_i(\vec{v})\geq 0$ is characterized by the fact that for any $\vec{v}_*\in\cone$ such that $\E_i(\vec{v}_*)=0$, we necessarily have $\E_j(\vec{v}_*)=0$ for some other $\E_j$ in $\mgs$. Algebraically, it is characterized by the fact that it is a \textit{positive} linear combination of other inequalities in the set. Thus, if $\E_i(\vec{v})\geq 0$ is redundant, then the system of inequalities obtained by ``ignoring'' $\E_i$ (i.e., the set of inequalities specified by $\mgs'=\mgs\setminus \{\E_i\}$) carves out the same  cone $\cone$. Given an arbitrary set of inequalities, one can use a linear program to check if a specific inequality is redundant. Indeed, there are efficient algorithms which reduce the set of inequalities to a minimum subset that specifies the same cone, where each inequality corresponds to a facet \cite{Fukuda15}. We therefore assume that each $\E_i\in\mgs$ corresponds to a non-redundant inequality.\footnote{\,This assumption is mainly made here for the purpose of defining the cone and its symmetries. However, we will introduce certain redundant inequalities in \S\ref{subsec:structure} to specify a region of the cone that we want to ignore. This is why at the end of the previous section we chose to keep the maximal elements of the MI-poset.}

\paragraph{Poset:}

Having described the geometric structure of our set-up, we now assume that we further have a partial order among the elements of $\mgs$, and denote by $\p=(\mgs,\preceq)$ the resulting poset. For each face $\face$ of $\cone$, consider the subset $\van(\face)\subseteq\mgs$ associated to $\face$ given by
\begin{equation}
\label{eq:van_set_of_face}
    \van(\face)=\{\E_i\in\mgs,\; \E_i(\vec{v})=0 \;\; \text{for all} \;\; \vec{v}\in\text{int}(\face)\} .
\end{equation}
This is the subset of $\mgs$ whose elements specify the inequalities which are saturated by the vectors in the interior of $\face$ (notice that this set is independent from the choice of vector $\vec{v}\in\text{int}(\face)$). If $\van(\face)$ is a \textit{down-set} in $\p$, we will say that $\face$ is a \textit{down-set face} (D-face). Our goal will be to find the 1-dimensional D-faces, or \textit{down-set extreme rays} (D-ERs), of $\cone$.\footnote{\,For the case where $\p$ is the MI-poset, D-faces and D-ERs are precisely the KC-faces and KC-ERs, respectively.}

\paragraph{Symmetries:}

Our next assumption regards the symmetries that might be present in this set-up. Given a polyhedral cone, the set of its faces ordered by reverse inclusion\footnote{\,We use reverse inclusion rather than inclusion because it will be more convenient for comparison with another lattice that will be introduced momentarily. This is also the convention adopted in \cite{He:2022bmi}.} forms a lattice, which we shall call the face lattice $\lat{\cone}$. Given \eqref{eq:van_set_of_face}, we can interpret each element of this lattice as the subset $\van(\face)$. Consider now the group $\text{Sym}(k)$, whose action on $\mgs$ consists of all possible permutations of $\mgs$, and also its natural inherited (element-wise) action on the power-set $2^{\mgs}$. To simplify the following discussion, we will henceforth not distinguish between these two actions, and it should always be clear from context which one we mean. In general, an arbitrary element of $\text{Sym}(k)$ does not preserve $\lat{\cone}$, even as a set.\footnote{\,For example, in \Cref{fig:3d_cone} the set $\lat{\cone}$ is not closed with respect to the action of $g=(23)\in\text{Sym(4)}$.} Therefore, we will focus on the largest subgroup $G_\cone\leq \text{Sym}(k)$ of permutations that do preserve $\lat{\cone}$, including its lattice structure, and call $G_\cone$ the \textit{combinatorial automorphism group of the cone} $\cone$ \cite{bremner2007}.\footnote{\,\label{ft:autom}More precisely, we demand that each $g\in G_\cone$ acts as a bijective \textit{homomorphism} (with respect to meet and join) of $\lat{\cone}$ into itself. As explained in detail in Appendix~\ref{appendix}, since $\mgs$ is finite, and $\lat{\cone}$ is the lattice of closed subsets of $\mgs$ under a closure operator (introduced in the next section), this requirement is effectively already implied by the assumption that the set $\lat{\cone}$ is closed with respect to the $g$ operation.} 

Since the geometric and order structures of our set-up are independent from each other, an element of $G_\cone$ in general might map a D-face to another face that is not a D-face. We assume now that the ``structure'' of the down-set corresponding to a D-face is important for the problem we want to solve with our algorithm,\footnote{\,Otherwise, there is no need to introduce the poset structure of $\mgs$ in the first place.} so we want to consider the group of transformations that preserve this structure. To make this precise, we specify in what sense two
down-sets can be considered structurally equivalent. The set of down-sets of $\p$, ordered by inclusion, forms a lattice $\lat{\p}$, which we shall call the \textit{down-set lattice} \cite{birkhoff1967lattice}. As for the case of $\lat{\cone}$, we then define the \textit{combinatorial automorphism group of the poset} $G_{\p}$ as the largest subgroup of $\text{Sym}(k)$ that preserves this lattice. Two down-sets of $\p$ will be considered equivalent if they are related by such a permutation.

In general, $\lat{\cone}$ and $\lat{\p}$, and therefore $G_\cone$ and $G_\p$, are unrelated. Since we want to preserve both lattices, we are then interested in the largest common subgroup\footnote{\,The \textit{set} $\lat{\cone\p}=\lat{\cone}\cap\lat{\p}$ also admits a lattice structure, since it is the set of closed subsets of $\mgs$ under the closure operator $\cl_{\text{FD}}$ (defined in the next section). However, in general $\lat{\cone\p}$ is not a sublattice of either $\lat{\cone}$ or $\lat{\p}$, though any $g\in G_{\cone\p}$ is still an automorphism of $\lat{\cone\p}$ for the same reason explained in \Cref{ft:autom}.}
\begin{equation}
\label{eq:group}
    G_{\cone\p}:=G_{\cone}\wedge G_{\p}.
\end{equation}
For generic polyhedral cones and posets, it may however not be so easy to determine this group. Thus, in our algorithm below, we will simply assume we know a subgroup $G\leq G_{\cone\p}$. The goal of the algorithm will then be to find all the D-ERs of $\cone$ only up to the action of $G$. Of course, there might exist a larger group $G'$, with $G \leq G' \leq  G_{\cone\p}$, and this would also preserve both the face and down-set lattices. Knowing $G'$ would naturally lead to an increased efficiency of our algorithm, so we will take $G$ as the largest such subgroup we can conveniently determine.

\paragraph{Excluded region:} 

Finally, we also consider the possibility that in the problem we want to solve with our algorithm, we are only interested in finding the extreme rays that do not belong to a specified collection $\regcone$ of faces of $\cone$. We will call $\regcone$ the \textit{excluded region} of $\cone$ and assume that it is closed under the action of $G$. Specifically, each face $\face\in\regcone$ can be described by its corresponding set of saturated inequalities $\van(\face)$, and using $\van(\face)$ it is straightforward to close its orbit under the action of $G$.\footnote{\,Notice that we are not making any assumption about the specific faces that appear in $\regcone$, and in general some of them can be facets. In this case it might be possible to simplify the problem by ignoring some of the inequalities to obtain a larger cone that is guaranteed to include all extreme rays of $\cone$ which are not in $\regcone$. We will further comment on this possibility in the next section, since in one variation of our algorithm we will apply this strategy to lower-dimensional faces generated during the computation.} While for a given specific problem the excluded region can be empty (no face is excluded), this notion, appropriately generalized to lower-dimensional faces, will play a key role during the computation. Since the algorithm proceeds by iteration, exploring at each step orbits of faces of various dimensionalities, we will use the excluded region to keep track of faces that have already been explored, or that are guaranteed to not include any new D-ER (more on this below), thereby allowing us to ignore them at subsequent steps of the computation.

\subsection{Closure operators}
\label{subsec:closures}

As mentioned before, our strategy to find the D-ERs of $\cone$ will consist of scanning over the entire power-set $2^{\mgs}$ while ignoring most subsets $\x\subseteq\mgs$ that are either not down-sets or do not correspond to faces.\footnote{\,In this subsection, we will ignore the symmetries and the excluded region $\regcone$; we will then combine all these ingredients in the following subsection.} Since the conditions for $\x$ to be a down-set or a face correspond to two completely different and unrelated structures, namely the cone and the poset, we need a sufficiently general tool that ``combines'' both of these conditions and make the scan over $2^{\mgs}$ more efficient. One such tool that will be key in our algorithm is the notion of a \textit{closure operator} \cite{Davey_Priestley_2002}.

\begin{defi}
\label{def:closure}
    Given any set $\mgs$, a closure operator on $\mgs$ is a map $\text{\emph{cl}}:2^{\mgs}\rightarrow2^{\mgs}$ such that, for any $\x,\y \subseteq \mgs$, we have
    \begin{enumerate}[label={\footnotesize \roman*)}]
        \item $\x \subseteq \text{\emph{cl}}(\x)$
        \item $\x \subseteq \y \, \implies \, \text{\emph{cl}}(\x) \subseteq \text{\emph{cl}}(\y)$
        \item $\text{\emph{cl}}(\text{\emph{cl}}(\x)) = \text{\emph{cl}}(\x)$.
    \end{enumerate}
\end{defi}
A subset $\x$ is said to be \textit{closed} if $\cl(\x)=\x$. Notice that from \Cref{def:closure}, we can immediately derive the implication
\begin{equation}
\label{eq:cl_implication}
    \x\subseteq\y\subseteq\cl(\x) \quad \implies \quad \cl(\y)=\cl(\x) .
\end{equation}
Indeed, taking the closure of all the terms on the l.h.s.\ of \eqref{eq:cl_implication} and using the second property in \Cref{def:closure}, we get $\cl(\x)\subseteq\cl(\y)\subseteq\cl(\cl(\x))$. The third property then implies $\cl(\x)\subseteq\cl(\y)\subseteq\cl(\x)$, from which \eqref{eq:cl_implication} follows. This observation suggests that if we are only interested in closed subsets, and we can efficiently compute $\cl(\x)$ for any $\x$, then we might be able to scan over $2^{\mgs}$ more efficiently if for most $\x$ 
there are ``many'' subsets $\y$ between $\x$ and $\cl(\x)$. 

In our set-up, we will consider three different closure operators.

\begin{enumerate}
    
    \item \emph{Down-set closure} ($\dc$): For any given $\x\subseteq\mgs$, we define
    \begin{equation}
    \label{eq:dc_defi}
        \dc(\x)=\{\E_i\in\mgs,\; \E_i\preceq \E_j\; \text{for some}\; \E_j\in\x \} .
    \end{equation}
    Clearly this is a closure operator according to \Cref{def:closure}.
    
    \item \emph{Linear-dependence closure} ($\lc$): For any given $\x\subseteq\mgs$, we define
    \begin{equation}
    \label{eq:lc_defi}
        \lc(\x)=\{\E_i\in\mgs,\; \text{rank}(\x\cup \{\E_i\}) = \text{rank}(\x) \} .
    \end{equation}
    Again, this is a closure operator according to \Cref{def:closure}.

    \item \emph{Face closure} ($\fc$): The last closure operator that we introduce is more subtle. We will give a geometric definition, and refer the reader to \cite{bjorner} for more details and the proof that this definition indeed gives a closure operator (although intuitively this should already be quite clear from the geometric description). Let $\hyp_i$ be the hyperplane consisting of the vectors $\vec v$ satisfying the equality $\E_i(\vec v) = 0$.
    For any given $\x\subseteq\mgs$, let $\mathbb{V}_{\x}$ be the linear subspace of $\mathbb{R}^n$ given by
    \begin{equation}
        \mathbb{V}_{_{\!\x}}=\bigcap_{\E_i\in\x} \hyp_i ,
    \end{equation}
    and consider the face $\face_{_{\!\x}}=\cone\cap\mathbb{V}_{_{\!\x}}$. Using \eqref{eq:van_set_of_face} we define
    \begin{equation}
    \label{eq:fc_defi}
        \fc(\x)=\van(\face_{_{\!\x}}) .
    \end{equation}
    In other words, if the highest dimensional face $\face_{_{\!\x}}$ contained in $\mathbb{V}_{_{\!\x}}$ only spans a proper subspace of $\mathbb{V}_{_{\!\x}}$, then the face closure $\fc(\x)$ adds further elements from $\mgs$ to reduce $\mathbb{V}_{_{\!\x}}$ to this subspace.
\end{enumerate}

Since we are neither interested in faces that are not D-faces nor in down-sets that do not correspond to faces, we want to compose these closure operators. Consider an arbitrary set $\mgs$ and two distinct closure operators $\cl_1$ and $\cl_2$. Notice that the composition $(\cl_1\circ\cl_2)$ satisfies the first two conditions in \Cref{def:closure}. Indeed, we have
\begin{equation}
    \x\subseteq\cl_2(\x)\subseteq\cl_1(\cl_2(\x))=(\cl_1\circ\cl_2)(\x),
\end{equation}
since both $\cl_1$ and $\cl_2$ are closure operators, thereby proving the first condition in \Cref{def:closure}. Similarly, 
\begin{align}
    \x \subseteq \y \, & \implies \, \cl_2(\x) \subseteq \cl_2(\y) \, \implies \, \cl_1(\cl_2(\x)) \subseteq \cl_1(\cl_2(\y))\nonumber\\
    & \implies \, (\cl_1\circ\cl_2)(\x) \subseteq (\cl_1\circ\cl_2)(\y),
\end{align}
which proves the second condition. However, $(\cl_1\circ\cl_2)$ is in general not a closure operator because it does not satisfy the third condition in \Cref{def:closure}, as $\cl_2(\x)$ is not necessarily a closed set with respect to $\cl_1$.\footnote{\,As a simple example, consider a poset $\p$ such that there is no element which is incomparable with every other element. Let $\cl_1=\dc$ as defined above, and $\cl_2=\cl_{\text{U}}$, where $\cl_{\text{U}}$ is the \textit{up-set} closure operator defined analogously to $\dc$ where we simply replace $\preceq$ with $\succeq$ in \eqref{eq:dc_defi}. Then the only subset $\x\subseteq\p$ which is closed under both $\cl_1$ and $\cl_2$ is $\x=\p$. } 

Consider now the operator $(\cl_1\circ\cl_2)^{\iota}$, defined by applying $\iota$ iterations of $(\cl_1\circ\cl_2)$. Each application of $(\cl_1\circ\cl_2)$ either augments the previous set by some element, or it preserves the set. In the latter case, further applications of $(\cl_1\circ\cl_2)$ do not do anything since the set is already closed, whereas in the former case, the maximum number of augmentations is obviously bounded by the cardinality of the set (although typically this bound is far from being saturated).  Therefore, if we define
\begin{align}
    \cl_{12}:\quad & 2^{\mgs} \rightarrow 2^{\mgs}\nonumber\\
        & \x \mapsto (\cl_1\circ\cl_2)^{|\mgs|}(\x) ,
\end{align}
it then follows that $\cl_{12}$ is a closure operator. In the next section, we will consider the following closure operators, obtained in this way by combining $\fc$ and $\lc$ with $\dc$, so that
\begin{align}
    \fdc & := (\fc\circ\dc)^{|\mgs|} \nonumber\\
     \ldc & := (\lc\circ\dc)^{|\mgs|}.
\end{align}
We will use the first operator $\fdc$ to formulate the algorithm in terms of faces of the cone. The second operator $\ldc$ will instead appear in the variation of the algorithm, discussed in \S\ref{subsec:variations}, which we actually used to derive the results presented in \S\ref{sec:five} and \S\ref{sec:results}.

Finally, let us make a brief comment about the relation between a closure operator and a group action that might encode the symmetries of a specific problem (as described in the previous section). Given an arbitrary closure operator $\cl$ on a set $\mgs$, the set of all subsets of $\mgs$ closed under $\cl$, which we denote as $\lat{\cl}$, is a lattice \cite{Davey_Priestley_2002}. Let $G$ be a subgroup of $\text{Sym}(|\mgs|)$ such that the set $\lat{\cl}$ is closed with respect to each $g\in G$, where $g$ acts element-wise. Then we prove in Appendix~\ref{appendix} that every $g\in G$ is a lattice homomorphism of $\lat{\cl}$ (i.e., it preserves meet and join), and for every $\x\subseteq\mgs$ (not necessarily closed) and $g\in G$, we have
\begin{equation}
\label{eq:commutation}
    \cl(g(\x)) = g(\cl(\x)).
\end{equation}
We are now ready to describe the algorithm in detail.

\subsection{The algorithm in detail}
\label{subsec:structure}

We first specify the main ingredients of the algorithm and list the key steps to orient the reader. Subsequently, when we prove the completeness of the algorithm (\Cref{thm:completeness}), we explain its logic and the rationale behind its construction.

The computation will be organized in terms of certain D-faces of $\cone$ which will be found by intermediate steps of the algorithm. However, since our goal is to compute just the D-ERs, each belonging to multiple faces,
only some D-faces will be determined by the algorithm (which D-faces in particular will depend on different choices made during the computation, devised such that the chosen collection contains all the D-ERs). Each D-face $\face$ will be associated to a triplet of the form $\mathfrak{T}=(\ds,\irr,\stab{}_\ds)$, where $\ds=\van(\face)$, $\stab{}_\ds$ is the subgroup of $G$ that \textit{stabilizes} $\ds$,\footnote{\,This is the \textit{stabilizer group} of $\ds$, i.e., the subgroup of $G$ of elements that leave $\ds$ unchanged.} and $\irr\subseteq\mgs$ is an \textit{up-set} specifying an excluded region $\regface$ on the boundary of $\face$ as follows. To each $\E_i\in\irr$, we associate the hyperplane 
\begin{equation}
\label{eq:ineq_hyper}
    \hyp_i=\{\vec{v}\in\mathbb{R}^n, \E_i(\vec{v})=0\} .
\end{equation}
This hyperplane specifies a face $\face_i$ on the boundary of $\face$ via 
\begin{equation}
\label{eq:face_from_ineq}
    \face_i=\hyp_i\cap\face.
\end{equation}
Given $\irr$, we then define the excluded region of $\face$ as the collection of faces
\begin{equation}
\label{eq:excluded_region_def}
   \regface=\{\face_i,\; \E_i\in\irr\}.
\end{equation}
Notice that while in \S\ref{subsec:set-up} we characterized a face in the excluded region of the cone $\regcone$ by the set $\van (\face)$, here we are instead specifying each element of $\regface$ (which is also a face of $\cone$) via the saturation of a single inequality. The equivalence of these descriptions will become evident momentarily, when we describe the initialization of the algorithm and explain how by introducing certain redundant inequalities we can translate the description of $\regcone$ into the language that we are using here.

The algorithm is initialized with one of these triplets. The \textit{main subroutine} then ``processes'' this triplet and outputs a set of new  triplets which are ``simpler'' in a sense that we will clarify below, and possibly some D-ERs. We then repeat this procedure until a certain criterion is met. This criterion is specified during the initialization and allows for some flexibility in how the algorithm can be used in practice. A schematic description of the algorithm is shown in Algorithm~\ref{alg:structure}, and we now describe each component in detail.

\begin{algorithm}[tb]
\caption{Global structure of the algorithm}
\label{alg:structure}
\BlankLine
\Input{the initial triplet $\mathfrak T^{(0)} = (\ds^{(0)},\irr^{(0)},\stab{(0)})$ \\
$f(\mathfrak{T},\lambda)$ in \Cref{cri:stop} used to stop the computation\\
$g(\mathfrak{T})$ in \Cref{cri:order} used by the main subroutine to generate new triplets}
\Output{a set $\textgoth{T}$ of triplets such that Condition~\ref{cond:triplet_conditions} holds and $f(\mathfrak{T},\lambda)=\text{TRUE}$\\
a set $\textgoth{R}$ of D-ERs}
\BlankLine
$\textgoth{T} \leftarrow \{\mathfrak{T}^{(0)}\}$\;
$\textgoth{R} \leftarrow \varnothing$\;
\While{$\exists\,\mathfrak{T}\in\textgoth{T}$ such that $f(\mathfrak{T},\lambda)=\text{\emph{FALSE}}$}{
    run the main subroutine (Algorithm~\ref{alg:main_subroutine}) on the first $\mathfrak{T}\in\textgoth{T}$ such that $f(\mathfrak{T},\lambda)=\text{FALSE}$\;
    update $\textgoth{T}$ by deleting $\mathfrak{T}$\;
    add to $\textgoth{T}$ the new triplets in the output of the main subroutine\;
    add to $\textgoth{R}$ any new D-ER in the output of the main subroutine\;
}
\end{algorithm}

\paragraph{Initialization:}  

The algorithm begins with a triplet $\mathfrak T^{(0)} = (\ds^{(0)},\irr^{(0)},\stab{(0)})$ corresponding to the set-up described in \S\ref{subsec:set-up}. The face $\face_0$ is the full cone $\cone$, which corresponds to $\ds^{(0)}=\varnothing$; $\stab{(0)}=G$ is the subgroup of $G_{\cone\p}$ that we assume to be known explicitly; and $\excluded$ corresponds to a new set of \textit{redundant} inequalities  which conveniently specifies the excluded region $\regcone$ of the cone, and is constructed as follows. As explained before, each face $\face_i\in\regcone$ is described by the set $\van(\face_i)$, which we can use to construct the dual vector 
\begin{equation}\label{eq:redundant}
    \E_i= \! \sum_{\E_j\,\in\,\van(\face_i)}\, \!\!\! \sigma_j \E_j  ,
\end{equation}
for some choice of $\sigma_j \in \mathbb R^+$. Since this is a positive linear combination, the inequality $\E_i(\vec{v})\geq 0$ is redundant, and it is saturated only by vectors that saturate all inequalities $\E_j(\vec{v})\geq 0$ for $\E_j\in \van(\face_i)$. Equivalently, this means that if we let $\mathbb{H}_i$ be the hyperplane specified by $\E_i$ as in \eqref{eq:ineq_hyper}, we have $\face_i=\mathbb{H}_i\cap\cone$. This is convenient because it allows us to treat the excluded region $\regcone$ of $\cone$ in the same way as we treat the excluded region $\regface$ of a D-face $\face$ obtained at an intermediate stage of the computation. In terms of the poset structure, there are many ways to extend the poset $\mgs$ to include these new dual vectors $\E_i$ (one for each $\face_i \in \regcone$). The simplest is to just have the redundant $\E_i$ form an antichain that is ``wholly disjoint'' from the rest of the poset,\footnote{\,By this we mean that each $\E_i$ is incomparable to any other element in the poset.} although it is oftentimes natural to add the redundant $\E_i$ ``at the top'' of the original poset.\footnote{\,By this we mean that each $\E_i$ covers some other element of the poset but is not covered by any element.} In this case, a simple solution that is guaranteed to preserve the symmetry of the problem is to specify cover relations such that each new $\E_i$ covers every maximal element of the original poset. However, this is not strictly necessary, and in specific situations other extensions might be more natural.\footnote{\,For example, in the case of the MI-poset, the only redundant inequalities are precisely the ones given by the entropies ``at the top'' of the poset, and they specify part of the excluded region.} Importantly, regardless of how we choose to extend the poset with the new dual vectors, while the full set of D-ERs of $\cone$ might change because of the extension, the set of D-ERs which are not in $\regcone$ remains unchanged.

In addition to the initial triplet $\mathfrak{T}^{(0)}$, the initialization of the algorithm also requires the specification of a function that will be used to stop the computation. Denoting by $\mathfrak{T}$ an arbitrary triplet generated during the computation, we define this function as follows.

\begin{fcn}
\label{cri:stop}
    A Boolean function $f(\mathfrak{T},\lambda)$, fixed for all triplets, which returns \emph{TRUE} if a chosen set of attributes of $\mathfrak{T}$ (which enter into the definition of $f$) satisfies a set of constraints parametrized by $\lambda$. 
\end{fcn}

Typically, $f$ will be a function that can be computed efficiently from the data given by a triplet, and should be chosen depending on the specific details of the problem to which the algorithm is applied. One useful example of this function is an upper bound to the dimension of the face corresponding to the triplet $\mathfrak{T}$, where $f$ returns TRUE if the bound is satisfied and FALSE otherwise. If one wants to use the algorithm to \textit{directly} find all $G$-orbits of D-ERs of $\cone$ which are not in $\regcone$, then one can simply set this bound to 1. Another option instead is to set a larger upper bound for the dimension of a face, and use the algorithm to restrict the search space to a point where one can then use other methods to find the extreme rays of the resulting lower-dimensional faces. An alternative option that might be useful is to stop the computation only when the excluded region on a face is sufficiently vast.\footnote{\,We will use the first option to obtain the results presented in \S\ref{sec:five} and a combination of the latter two to obtain the results presented in \S\ref{sec:results}.}

Finally, at this stage one should also choose a function $g$ to specify another criterion (see \Cref{cri:order} below), which will be used at each iteration by the main subroutine to generate new triplets given one triplet as input. As we will explain, the choice of this function affects which D-faces are found by the algorithm at intermediate steps of the computation, and different choices might make a significant difference in terms of efficiency.

\begin{algorithm}[tb]
\caption{The main subroutine}
\label{alg:main_subroutine}
\BlankLine
\Input{a triplet $\mathfrak{T}=(\ds,\irr,\stab{}_\ds)$ that satisfies Condition~\ref{cond:triplet_conditions}}
\Output{a set $\ddot{\textgoth{T}}$ of triplets that satisfy Condition~\ref{cond:triplet_conditions}\\
a set $\textgoth{R}$ of D-ERs}
\BlankLine
$\orb \leftarrow \text{the maximal elements of}\; \mgs\setminus (\ds\cup\irr)$\;
$\mathfrak{M} \leftarrow \text{the partition of}\;\orb$ into $\stab{}_\ds$-orbits\;
$\mathfrak{D} \leftarrow \varnothing$\;
\For{each orbit $\mathcal{M}_m\in\mathfrak{M}$}
{
    choose an arbitrary element $\E_m$ of the orbit $\mathcal{M}_m$\; 
    $\ds_{m} \leftarrow \fdc(\ds\cup \{\E_m\})$\;
    append $\ds_{m}$ to $\mathfrak{D}$\;
    }
\nl $\overline{\irr} \leftarrow \text{the union of}\;\; \irr\;\; \text{and all orbits}\;\, \mathcal{M}_m\; \text{of}\;\;  \mathfrak{M}\;\; \text{such that}\;\;  \ds_m\cap\,\irr\,\neq\,\varnothing$\;
 $\dot{\mathfrak{D}} \leftarrow \{\ds_{m}\in\mathfrak{D},\, \ds_{m}\cap\irr = \varnothing\}$\;
$\dot{\mathfrak{M}} \leftarrow \{\mathcal{M}_{m}\in\mathfrak{M},\, \ds_{m}\cap\irr = \varnothing\}$\;
 \For{each $\mathcal{D}_m\in\dot{\mathfrak{D}}$}
{
    \If{the face $\face_m$ is 1-dimensional}
    {add $\face_m$ to $\textgoth{R}$\;
    delete $\ds_m$ from $\dot{\mathfrak{D}}$\;
    add all elements of $\mathcal{M}_m$ to $\overline{\irr}$\;
    delete $\mathcal{M}_m$ from $\dot{\mathfrak{M}}$\;
    }
    }
\nl $(\ddot{\mathfrak{D}},\ddot{\mathfrak{M}}) \leftarrow$ the output of the subroutine in Algorithm~\ref{alg:permutations} run on $(\dot{\mathfrak{D}},\dot{\mathfrak{M}})$\;
\nl $\textgoth{T} \leftarrow$ the output of the subroutine in Algorithm~\ref{alg:triplets} run on $(\ddot{\mathfrak{D}},\ddot{\mathfrak{M}})$\;
\nl $\dot{\textgoth{T}} \leftarrow$ the output of the subroutine in Algorithm~\ref{alg:update} run on $\textgoth{T}$\;
\nl append to $\dot{\textgoth{T}}$ the triplet $(\ds,\irr\cup\mathcal{M},G_{\ds})$\;
$\ddot{\textgoth{T}} \leftarrow \varnothing$\;
\nl \For{each $\mathfrak{T}\in\dot{\textgoth{T}}$}
{
    \eIf{$\text{\emph{dim}}(\face) - \text{\emph{rank}}
    \pqty{\eval{\mgs\setminus\!(\ds\cup\irr)}_{\face}}
    =0$}
    {add $\mathfrak{T}$ to $\ddot{\textgoth{T}}$\;}
    {\If{$\text{\emph{dim}}(\face) - \text{\emph{rank}}
    \pqty{\eval{\mgs\setminus\!(\ds\cup\irr)}_{\face}}
    =1$}
   { \nl if $\mgs\setminus\irr=\van(\face)$ for some $\face\notin\regcone$, add $\face$ to $\textgoth{R}$\;}
    }
    }
\end{algorithm}

\paragraph{Main subroutine:} 

The main subroutine is the core of our algorithm. The input is a triplet $\mathfrak{T}=(\ds,\irr,\stab{}_\ds)$, corresponding to a face $\face$, that satisfies \Cref{cond:triplet_conditions} outlined below. During the computation, new triplets will be generated which do not necessarily satisfy these conditions. Some of these triplets will correspond to D-ERs and will be stored, while the others will be either discarded or modified such that the conditions are satisfied. The triplets in the output correspond to a simplification of the problem because for each one of them, either the down-set strictly contains $\ds$, and therefore the dimension of the corresponding face is smaller than that of $\face$, or the excluded region is strictly larger, therefore reducing the set of lower-dimensional faces that we still need to consider in our search of D-ERs. Furthermore, most new triplets have both of these properties. 

We now describe the conditions that we require for each triplet and the corresponding face. By our assumptions about the set-up described in \S\ref{subsec:set-up}, it is trivial to check that the triplet $\mathfrak{T}^{(0)}$ used to initialize the algorithm satisfies these conditions.\footnote{\,This is true except for trivial situations, such as the case where $\cone$ is 1-dimensional, which we ignore.}
\begin{condition} 
\label{cond:triplet_conditions}
Given a triplet $\mathfrak{T}=(\ds,\irr,\stab{}_\ds)$ corresponding to a face $\face$ and excluded region $\regface$, we require
\begin{align*}
         \text{{\footnotesize i)}}\quad & \ds\cap\irr = \varnothing \qquad
        && \text{{\footnotesize iii)}}\quad \irr\; \text{is stabilized by}\;\, \stab{}_\ds \\
        \text{{\footnotesize ii)}}\quad & \text{\emph{dim}}(\face)>1
        && \text{{\footnotesize iv)}}\quad \text{\emph{dim}}(\face) = 
        \text{\emph{rank}}\,\big(\mgs\setminus\!(\ds\cup\irr)\big|_{\face}\big) 
\end{align*}
\end{condition}

The first condition ensures that the excluded region $\regface$ on $\face$ is not the entire face. The second condition demands that $\face$ is not a D-ER. The third condition ensures that if a D-ER is in the excluded region $\regface$, then any other D-ER in the same $\stab{}_\ds$-orbit is also excluded. Finally, the last condition deserves a more careful explanation. Suppose we have a D-face $\face$, with $\van(\face)=\ds$, and an excluded region $\regface$ specified by $\irr$. Since we are only interested in finding D-ERs in $\face$ outside of $\regface$, we only care about ERs that do not saturate any of the inequalities specified by $\irr$. We denote by $\eval{\mgs\setminus\!(\ds\cup\irr)}_{\face}$
the reduction of the inequalities specified by the set  $\mgs\setminus\!(\ds\cup\irr)$ to the linear subspace spanned by $\face$. The last condition then demands that the cone specified by these inequalities in the subspace spanned by $\face$ is a pointed cone.

\begin{algorithm}[tb]
\caption{Checking the equivalence of faces under $G_{\ds}$}
\label{alg:permutations}
\BlankLine
\Input{a collection $\dot{\mathfrak{D}}$ of down-sets, each one corresponding to a face\\
a collection $\dot{\mathfrak{M}}$ of subsets of $\mathcal{M}$, one for each down-set in $\dot{\mathfrak{D}}$\\
the group $G_{\ds}$ from the computation in the main subroutine}
\Output{a subset $\ddot{\mathfrak{D}}$ of $\dot{\mathfrak{D}}$ whose elements belong to distinct $G_{\ds}$-orbits\\
a new collection $\ddot{\mathfrak{M}}$ of subsets of $\mathcal{M}$, one for each down-set in $\ddot{\mathfrak{D}}$\\
}
\BlankLine
$\ddot{\mathfrak{D}} \leftarrow \varnothing$\;
$\ddot{\mathfrak{M}} \leftarrow \varnothing$\;
$\textgoth{O}\leftarrow$ the partition of $\dot{\mathfrak{D}}$ into $G_{\ds}$-orbits\;
\For{each orbit in $\textgoth{O}$}
{append to $\ddot{\mathfrak{D}}$ an arbitrary representative of the orbit\;
append to $\ddot{\mathfrak{M}}$ the union of all $\mathcal{M}_m$ in $\dot{\mathfrak{M}}$ corresponding to the down-sets $\mathcal{D}_m$ in the orbit\;
}
\end{algorithm}

The computation of the main subroutine is presented in full in Algorithm~\ref{alg:main_subroutine}. Here we clarify some of the steps, describe the internal subroutines presented in Algorithm~\ref{alg:permutations}, Algorithm~\ref{alg:triplets}, and Algorithm~\ref{alg:update}, and present a few possible choices for the function $g$ which affects the output at each iteration. The numbers in the numbered list below refer to corresponding lines in Algorithm~\ref{alg:main_subroutine}.

\begin{enumerate}[label={\scriptsize \textbf{\arabic*})}]
    \item After constructing $\overline{\irr}$, one may wonder if the efficiency of the algorithm can be improved by further checking whether there is some $\ds_m\in\dot{\mathfrak{D}}$ such that $\ds_m\cap\overline{\irr}\neq\varnothing$ while $\ds_m\cap\irr=\varnothing$, in which case one could add more elements to $\overline{\irr}$ and extend the excluded region for the new triplets that will be constructed in the following steps. However, this is not the case due to the following implication:
    \footnote{\,We can prove this as follows. Let $\ds_m = \fdc(\ds \cup \{\E_m\})$ for some $\E_m\in\mathcal{M}_{m}$, and suppose $\ds_m$ contains an element $\sf{U}\in\mathcal{M}_{m'}$ for some $\mathcal{M}_{m'}\subseteq \overline{\irr} \setminus \irr$. From the algorithm, this means that we have constructed some $\ds_{m'}=\fdc(\ds\cup \{\E_{m'}\})$ with $\E_{m'}\in\mathcal{M}_{m'}$ such that there is some $\sf{U}'\in\ds_{m'}\cap\irr$. By the defining properties of a closure operator, $\ds_m=\fdc(\ds\cup \{\E_m,\sf{U}\}) \supseteq \fdc(\ds\cup \{\sf{U}\})$, and since $\mathcal{M}_{m'}$ is an orbit of $\mathcal{M}$ generated by $G_{\ds}$, there is some $g\in G_{\ds}$ such that $g(\sf{U})=\E_{m'}$. To show that $\ds_m\cap\irr\neq\varnothing$,  it is sufficient to show that $g(\ds_m\cap\irr)\neq\varnothing$. Since $\irr$ is invariant under $G_{\ds}$, we have $g(\ds_m\cap\irr)=g(\ds_m)\cap\irr$, and by the above inclusion, $g(\ds_m)\cap\irr\supseteq g(\fdc(\ds\cup \{\sf{U}\}))\cap\irr$. Furthermore, since $g$ and $\fdc$ commute (cf.\ \eqref{eq:commutation}), and $\ds$ is invariant under $g$, we obtain $g(\ds_m\cap\irr)\supseteq \fdc(\ds\cup g(\{\sf{U}\}))\cap\irr = \fdc(\ds\cup \{\E_{m'}\}\cap\irr = \ds_{m'}\cap\irr \ni \sf{U}'$.
    } 
    \begin{equation}
        \ds_m\cap\,\irr=\varnothing\quad \implies\quad \ds_m\cap\,\overline{\irr}=\varnothing\qquad \text{for all} \;\;\ds_m\in\dot{\mathfrak{D}} .
    \end{equation}

    \item The goal of Algorithm~\ref{alg:permutations} is to check if some of the faces corresponding to the triplets in $\dot{\mathfrak{D}}$ are related by an element of $G_\ds$. 
    (Although this would not be the case for the generating elements $\mathcal{M}_m$ themselves by construction, the enlargement from using $\fdc$ can subsequently enhance the symmetry.)
    In that case, there is no reason to further process each equivalent triplet, and we simply select one. Furthermore, for each value of the index $m$, we combine all elements of each $\mathcal{M}_{m'}$ with $\ds_{m'}\sim\ds_m$ into a single new set. The resulting subsets of $\mathcal{M}$ will be used to generate the excluded regions of the new faces, as explained below.

\begin{algorithm}[tb]
\caption{Generating the triplets}
\label{alg:triplets}
\BlankLine
\Input{the collection $\ddot{\mathfrak{D}}$ of down-sets from the output of Algorithm~\ref{alg:permutations}\\
the collection $\ddot{\mathfrak{M}}$ of subsets of $\mgs$ from the output of Algorithm~\ref{alg:permutations}\\
the set $\overline{\irr}$ from the computation in the main subroutine\\
the \Cref{cri:order} chosen at initialization 
}
\Output{a collection $\textgoth{T}$ of triplets}
\BlankLine
$\textgoth{T} \leftarrow \varnothing $\;
reorder the elements of $\ddot{\mathfrak{D}}$ and $\ddot{\mathfrak{M}}$ using \Cref{cri:order}  (see main text)\; 
index the elements of these sequences by $q\in \{1,2,\ldots\,Q\}$
(this fixes $\ds_q$ and $\mathcal{M}_q$)\;
$\mathcal{M}_0 \leftarrow \varnothing$\;
$\irr_0 \leftarrow \overline{\irr}$\;
\For{each $q$}{
    $\irr_q \leftarrow \irr_{q-1}\cup\mathcal{M}_{q-1}$\;
    compute the stabilizer group $\stab{}_{\ds_q}$ of $\ds_q$\;
    append the triplet $\mathfrak{T}_q=(\ds_q,\irr_q,\stab{}_{\ds_q})$ to $\textgoth{T}$\;
    }
\end{algorithm}

    \item Algorithm~\ref{alg:triplets} generates a collection of new triplets starting from the output of Algorithm~\ref{alg:permutations}. It relies on the function $g(\mathcal{D}_m,\mathcal{M}_m)$, which has to be chosen at the initialization, to reorder the elements of $\ddot{\mathfrak{D}}$ and $\ddot{\mathfrak{M}}$. 
    \begin{fcn}
    \label{cri:order}
        A vector-valued function $g(\mathcal{D}_m,\mathcal{M}_m)$, which is fixed for all pairs and can be easily computed.
    \end{fcn}
    The algorithm orders the elements $\mathcal{D}_m\in\ddot{\mathfrak{D}}$ and $\mathcal{M}_m\in\ddot{\mathfrak{M}}$ such that the value of the first component of the output of \Cref{cri:order} is non-decreasing along the sequence. It then orders the elements $\mathcal{D}_m$ and $\mathcal{M}_m$ with the same value of the first component of \Cref{cri:order} according to the second component (again such that this is non-increasing), and so on. Depending on the details of the problem, different choices might have a decisive effect on the efficiency of the algorithm. Examples of the attributes computed by $g(\mathcal{D}_m,\mathcal{M}_m)$ can be the dimension of the corresponding face, or the cardinality of the corresponding down-set.
    
    \item Notice that at this stage of the main subroutine each triplet in the output of Algorithm~\ref{alg:triplets} already satisfies (i) and (ii) in \Cref{cond:triplet_conditions}, since the former is implemented by the definition of $\dot{\mathfrak{D}}$, and the latter by saving 1-dimensional faces into $\textgoth{R}$ and deleting the corresponding down-sets from $\dot{\mathfrak{D}}$ even before running Algorithm~\ref{alg:permutations}. Algorithm~\ref{alg:update} first implements (iii) in \Cref{cond:triplet_conditions} for each triplet by completing the orbit of $\irr_q$ under the action of $\stab{}_{\ds_q}$. It then attempts to ``extend'' the excluded regions by checking which elements in $\mgs\setminus(\ds_q\cup\irr_q)$ would, after adding them to $\ds_q$ and taking the closure, generate an element in $\irr_q$. By doing so using orbits of $\stab{}_{\ds_q}$, it is guaranteed that the new updated triplets in the output also satisfy (iii) in \Cref{cond:triplet_conditions}.

    \item The main goal of the main subroutine is to generate new triplets by adding at least one new element from $\mathcal{M}$ to $\ds$. However, we should also take into account the case where no such element is added to $\ds$. We can then add all the elements of $\mathcal{M}$ to the excluded region.

\begin{algorithm}[tb]
\caption{Updating the triplets}
\label{alg:update}
\BlankLine
\Input{the collection $\textgoth{T}$ of triplets from the output of Algorithm~\ref{alg:triplets}}
\Output{a collection $\dot{\textgoth{T}}$ of triplets that satisfies (iii) in Condition~\ref{cond:triplet_conditions}}
\BlankLine
$\dot{\textgoth{T}} \leftarrow \varnothing$\;
\For{each triplet $\mathfrak{T}_q=(\ds_q,\irr_q,\stab{}_{\ds_q})$ in $\textgoth{T}$}{
    update $\irr_q$ by completing its orbit under the action of $\stab{}_{\ds_q}$\;
   partition $\mgs\setminus(\ds_q\cup\irr_q)$ into $\stab{}_{\ds_q}$ orbits\;
   \For{each orbit}{
   choose a representative $\E_k$\;
   \If{$\text{\emph{cl}}_{\text{\emph{FD}}}(\ds_q\cup \{\E_k\})\cap\irr\neq\varnothing$}{
      add all the elements of the orbit to $\irr_q$}
    }
    }
\end{algorithm}

    \item Each triplet in $\dot{\textgoth{T}}$ satisfies the first three requirements of \Cref{cond:triplet_conditions}. At this step, we check which triplets in $\dot{\textgoth{T}}$ satisfy the last requirement, and include them as part of the final output $\ddot{\textgoth{T}}$.
    
    \item If a triplet in $\dot{\textgoth{T}}$ does not satisfy (iv) in \Cref{cond:triplet_conditions}, we have to distinguish two cases. If
    \begin{equation}
    \label{eq:1d-flat-condition}
        \text{dim}(\face) - \text{rank}\,(\left.\mgs\setminus\!(\ds\cup\irr)\right|_{\face}) = 1 ,
    \end{equation}
    it is possible that by further saturating all the inequalities specified by the elements of $\mgs\setminus\irr$ we obtain a subspace spanned by a D-ER which is not in $\regcone$, and in that case we add it to the output. On the other hand, if
    \begin{equation}
    \label{eq:non-pointed}
        \text{dim}(\face) - \text{rank}\,(\left.\mgs\setminus\!(\ds\cup\irr)\right|_{\face}) > 1 ,
    \end{equation}
    the inequalities specified by the elements of $\mgs\setminus\!(\ds\cup\irr)$ are not enough to carve out a pointed cone in the subspace spanned by $\face$. This means that any ER on $\face$ is necessarily in $\regcone$, and we can discard the triplet entirely.
\end{enumerate}

\paragraph{Proof of completeness:}

We now prove that for an arbitrary set-up as described in \S\ref{subsec:set-up}, our algorithm gives all the D-ERs which are not in the excluded region. 

\begin{thm}
\label{thm:completeness}
    For each down-set extreme ray $\vec{\face}$ of the cone $\cone$ which is not in the excluded region $\regcone$, at least one of the following is true:
    \begin{itemize}
        \item The set $\textgoth{R}$ contains at least one element of the $G$-orbit of $\vec{\face}$.
        \item The set $\ddot{\textgoth{T}}$ contains at least one higher-dimensional face $\face$ such that at least one representative of the $G$-orbit of $\vec{\face}$ is in $\face$ and not in $\regface$.
    \end{itemize}
\end{thm}
\begin{proof}
    As we mentioned before, the basic idea of the algorithm is to essentially scan over the entire space of all possible down-sets of the MI-poset and check which down-sets are D-ERs. To make this scan more efficient, we will use various constraints to ignore most down-sets, and we will prove below that these constraints do not miss any relevant ones. Since the full algorithm (see Algorithm~\ref{alg:structure}) is simply an iteration of the main subroutine in Algorithm~\ref{alg:main_subroutine}, we only need to show the statement of the theorem for a single iteration, i.e., we only need to show that given a triplet $\mathfrak{T}$ corresponding to a face $\face$ and an arbitrary D-ER $\vec{\face} \in \face$ not in the excluded region $\regface$ of $\face$, the output of the main subroutine includes at least one element in the $G_{\ds}$-orbit of $\vec{\face}$.

    We begin with setting $\ds = \van(\face)$, and notice that for any D-ER $\vec{\face} \in \face$, we have $\van(\vec{\face}) \supseteq \ds$. This means our algorithm should scan over all possible $\ds'$ such that $\ds' \supseteq \ds$. However, since we are not interested in any ER in $\regface$, we can impose the constraint $\irr \cap \ds' = \varnothing$, so that it suffices to consider all possible deformations of $\ds$ obtained by adding to it elements from an arbitrary subset of $\mgs \setminus (\ds \cup \irr)$. We will do so element by element, taking advantage of closure operators, the symmetries of the problem, and the poset structure. 
    
    Importantly, we do not consider all these possibilities in a single iteration of the main subroutine. Instead, to take advantage of the poset structure, in a single iteration we only consider the down-sets obtained by adding to $\ds$ an element of $\mathcal{M}$, which is the set of maximal elements of $\mgs\setminus(\ds\cup\irr)$. More precisely, we organize all possible down-sets $\ds'$ into two families: Those that include at least one element of $\mathcal{M}$, and those that do not. The down-sets in the second family are collected in a new triplet $(\ds,\irr\cup\mathcal{M},G_{\ds})$ (see line {\scriptsize {\bf 5}} in Algorithm~\ref{alg:main_subroutine}) and will be processed at a later iteration. Therefore, we only need to show that the procedure in Algorithm~\ref{alg:main_subroutine} guarantees a complete scan over the first family.

    Consider an element $\E \in \mathcal{M}$ and the set $\ds \cup \{\E\}$. If there is another element $\E'$ which is mapped to $\E$ by an element of the stabilizer group $G_{\ds}$ of $\ds$, then clearly $\ds \cup \{\E\}$ and $\ds \cup \{\E'\}$ are related by the same group operation.\footnote{\,Recall that the action of $G$ on subsets of $\mgs$ is defined element-wise, starting from that on $\mgs$.} Therefore, since we are only interested in finding D-ERs up to equivalence under the action of $G_{\ds}$, once we include in the scan the down-sets obtained from $\ds$ by adding $\E$, we can ignore any down-set obtained from $\ds$ by adding $\E'$. We then partition $\mathcal{M}$ into $G_{\ds}$-orbits and focus only on deformations of $\ds$ generated by representatives of these orbits. This partition of $\mathcal{M}$ is denoted by $\mathfrak{M}$.

    The set $\ds\cup\{\E\}$ does not necessarily correspond to a new D-face of the cone, and to obtain such a face we compute $\fdc(\ds\cup\{\E\})$. As explained in \S\ref{subsec:closures}, we can then ignore any down-set $\ds''$ such that $\ds\cup\{\E\} \subset \ds'' \subset \fdc(\ds\cup\{\E\})$ (cf.\ \eqref{eq:cl_implication}). Computing the closure for an arbitrary representative from every component in $\mathfrak{M}$, we obtain a set $\mathfrak{D}$ of down-sets, each of which corresponding to a D-face. Up to this point, it should be clear that the procedure guarantees completeness, in the sense that any D-ER of $\face$ not in $\regface$ is either (i) in one of these faces corresponding to a down-set in $\mathfrak{D}$; (ii) in another D-face that is related to a down-set in $\mathfrak{D}$ by the action of $\stab{}_{\ds}$; or (iii) in a D-face whose down-set does not include any element of $\mathcal{M}$ (and will therefore be processed at a later iteration). In particular, notice that (ii) follows from \eqref{eq:commutation}, which as explained in Appendix~\ref{appendix}, is implied by the definition of $G$ described in \S\ref{subsec:set-up}. The rest of the proof then focuses on demonstrating that the additional steps in the main subroutine used to increase the efficiency of our algorithm, namely ignoring certain elements of $\mathfrak{D}$ and producing the triplets (including new excluded regions) for the remaining ones, do not miss any relevant D-ERs.
 
    Starting from $\irr$, we now construct a new set $\overline{\irr}$ of elements that we will use to specify the excluded region of the new faces. First, notice that a new face corresponding to an element in $\mathfrak{D}$ might already be entirely contained in the excluded region $\regface$ of the initial face $\face$, i.e., if $\fdc(\ds\cup\{\E\}) \cap \irr \neq \varnothing$, in which case they can be ignored. If this is true, we add all elements in the orbit of $\E$ to $\overline{\irr}$, since we do not want to add $\E$ (or any equivalent element) to the down-set of any lower-dimensional face that we derive from $\face$. Furthermore, some elements of $\mathfrak{D}$ might already correspond to D-ERs, in which case we add them to $\textgoth{R}$ and again add all the elements in the orbit of $\E$ to $\overline{\irr}$, since adding any further element beyond $\E$ would only yield the trivial dimensional space. After these steps, we are left with a sub-collection of D-faces $\dot{\mathfrak{D}}$ and orbits $\dot{\mathfrak{M}}$ of elements of $\mathcal{M}$. Notice that all the faces corresponding to the elements of $\dot{\mathfrak{D}}$ now satisfy (i) and (ii) in \Cref{cond:triplet_conditions}.
    
    We remark that since the closure generally adds more elements to $\ds \cup \{\E\}$, it can also happen that even though another element $\E'$ is in a different orbit of $\mathcal{M}$, $\fdc(\ds \cup \{\E'\})$ is related to $\fdc(\ds \cup \{\E\})$ by an element of $G_{\ds}$, in which case the two corresponding faces are considered equivalent. This equivalence is checked by Algorithm~\ref{alg:permutations}, allowing us to extract an even smaller collection of representative faces $\ddot{\mathfrak{D}}$, and to partition the elements of $\mathcal{M}$ that have not been added to $\overline{\irr}$ into even larger sets (the elements of $\ddot{\mathfrak{M}}$), each set being the union of elements of $\dot{\mathfrak{M}}$. 

    To summarize, we have thus far used the closure operator $\fdc$ and the symmetries of the problem to identify a representative for each possible orbit of D-faces on the boundary of $\face$ that contain at least one element of $\mathcal{M}$ in their respective down-sets. The set of such representatives is $\ddot{\mathfrak{D}}$. We have further organized the elements of $\mathcal{M}$ into two sets: those generating new faces which we can ignore are in $\overline{\irr}$, and those in $\ddot{\mathfrak{M}}$ where each component of this partition is associated with one element of $\ddot{\mathfrak{D}}$. We will now describe how to use this data to generate new triplets (see Algorithm~\ref{alg:triplets}) by specifying an excluded region for each face corresponding to a down-set in $\ddot{\mathfrak{D}}$. 
    
    First, we reorder the elements of $\ddot{\mathfrak{D}}$ according to the function $g$ in \Cref{cri:order}. It is clear that if we specify the excluded region for each face in $\ddot{\mathfrak{D}}$ to be just $\overline{\irr}$, we are guaranteed to not miss any $\vec{\face} \in \face$ which is not in $\regface$ because such D-ERs do not include any element of $\overline{\irr}$ in their respective down-sets; however, this would be very inefficient. Instead, only the first triplet will have the excluded region be $\overline{\irr}$, and this triplet will be the input into the next iteration of the main subroutine. As for the second triplet, we can now include in its excluded region not only $\overline{\irr}$, but also all the elements in the first component of $\ddot{\mathfrak{M}}$ since any D-ER associated to a down-set containing such an element will be obtained when we input the first triplet into the main subroutine during a later iteration. In a similar manner, we can for the third triplet include in its excluded region the union of $\overline{\irr}$ and the first two components of $\ddot{\mathfrak{M}}$, and so on. Thus, we are able to enlarge the excluded region beyond $\overline{\irr}$ for all except the first triplet, thereby increasing the efficiency of the algorithm without missing any D-ER $\vec\face \in \face$ not belonging in $\regface$. We stress that the specific choice of function $g$ only affects the efficiency of the algorithm, and not its completeness. In other words, our proof that the algorithm is complete is independent of this function.

    We then use Algorithm~\ref{alg:update} to enforce (iii) in \Cref{cond:triplet_conditions} and to enlarge the excluded region of each new triplet even more by simply adding any element whose inclusion would, under the closure operator, cause the further addition of an element in the excluded region. Again by \eqref{eq:commutation}, it is clear that this does not affect completeness, and we can efficiently implement this enlargement by again organizing these elements into orbits under the action of the stabilizer group of the triplet, such that (iii) in \Cref{cond:triplet_conditions} remains satisfied. 
    
    The final step in the main subroutine is to check (iv) in \Cref{cond:triplet_conditions}. By following the procedure given in the description of line {\scriptsize {\bf 7}} in Algorithm~\ref{alg:main_subroutine}, we ensure that we do not miss any D-ERs associated to triplets that fail this condition. This completes the proof.

\end{proof}

Of course, in the case where the function $f$ in \Cref{cri:stop} is chosen such that at the end of the computation $\ddot{\textgoth{T}}$ is non-empty, one would then employ other, possibly more efficient at that stage, methods to extract the D-ERs. Furthermore, notice that for some D-ERs, the output might contain multiple representatives of the same $G$-orbit. We will comment on these issues in the next paragraph.

\paragraph{Post-processing:} 

When the set-up is sufficiently simple such that the algorithm is stopped only when there are no more triplets to be processed, all the desired D-ERs will be contained in the set $\textgoth{R}$. However, it should be noticed that in general $\textgoth{R}$ might contain more than one representative for each $G$-orbit, and to get exactly one representative per $G$-orbit one should still check equivalence of the various D-ERs under the action of $G$. 

For harder problems, it might instead be convenient to stop the computation earlier, when some of the triplets in $\ddot{\textgoth{T}}$ correspond to higher-dimensional faces. Typically, this is the case when for each triplet that remains to be processed, the sub-poset induced by $\mgs\setminus(\ds\cup\irr)$ is ``approximately an antichain'', or more precisely, when the down-set of each maximal element of this sub-poset is a small subset of the sub-poset. In this situation, it is typically more efficient to use a standard conversion algorithm to extract all extreme rays of a face, and then check which extreme rays are D-ERs. 

Nevertheless, we stress that even when using standard conversion algorithms, it might still be possible to speed up this computation significantly using the set $\irr$ that specifies the excluded region.\footnote{\,This trick was indeed crucial for the derivation of the results presented in \S\ref{sec:results}.} To see this, consider a triplet $(\ds,\irr,\stab{}_\ds)$ for a face $\face$, and suppose that $\irr$ is a ``large subset'' of $\mgs$, i.e., that $|\irr|/|\mgs| \sim 1 $. On the subspace spanned by $\face$, which is determined by $\ds$, $\face$ is a pointed cone carved out by the inequalities corresponding to the elements of $\mgs\setminus\ds$. On this subspace, consider the larger cone $\Check{\face}$ carved out by the inequalities given by $\mgs\setminus(\ds\cup\irr)$. Notice that since the triplet we are considering is in the output of the algorithm, it satisfies in particular (iv) in Condition~\ref{cond:triplet_conditions}, implying that this cone is pointed. While the ERs of $\Check{\face}$ are in general different from those of $\face$, notice that any ER of $\face$ which does not saturate any inequality from $\irr$ (i.e., that is not in the excluded region $\regface$) is also an ER of $\Check{\face}$. Since $\Check{\face}$ is specified by a much smaller set of inequalities compared to $\face$, finding its ERs using standard conversion algorithm can be much simpler.

Suppose now that this is the case, that we have performed this computation and have found all the ERs of $\face$ outside of $\regface$. We now need to extract the subset of D-ERs. Typically, one finds the ERs in a basis for the subspace of $\face$ rather than that for the full space. Since the poset $\mgs$ is defined in the full space where $\cone$ lives, in order to check if an ER is a D-ER, one should first perform a change of basis to recover the expression of the ER in the full space. However, if the number of both ERs and dimensions is large, this is not the most convenient strategy. Instead, one can pre-compute a set of necessary conditions that must be satisfied by the set of inequalities which are saturated by an ER for it to be a D-ER, and then check these conditions ``locally'' on the lower-dimensional subspace spanned by $\face$. These conditions essentially stem from the fact that if an inequality corresponding to an element $\E_i\in\mgs$ is saturated, all inequalities corresponding to the elements $\E_j\in\mgs$ with $\E_j\prec\E_i$ also have to be saturated. In the full space, this last requirement is obviously also sufficient (since it is the definition of a down-set), but it may not be verifiable in the lower-dimensional subspace. The reason is essentially the fact that in general several inequalities which are distinct in the full space are indistinguishable in the lower-dimensional space. Nonetheless, one can derive local necessary conditions by restricting this requirement to the set of dimensionally reduced inequalities. If $\mgs$ is far from being an antichain, which as discussed above is the regime of utility of this algorithm, one can expect that in most situations the vast majority of the ERs of $\face$ will already fail these simpler necessary conditions and can therefore be more efficiently discarded as candidate D-ERs. It then suffices to recover the full expression of a much smaller collection of ERs in the space of $\cone$ to finally verify which ERs are D-ERs using the complete poset structure.

\subsection{A few possible variations}
\label{subsec:variations}

We now describe a few possible variations of the algorithm. We will use some of them to derive the results presented in \S\ref{sec:five} and \S\ref{sec:results}.

\paragraph{Checking symmetries:}
Notice that in Algorithm~\ref{alg:structure}, the main subroutine is iterated on \textit{all} outputs of previous iterations until the criterion to stop the computation is met. Some triplets that are generated during the computation however might belong to the same $G$-orbit. One possible variation of the algorithm would check this equivalence and select a representative for each orbit, deleting other triplets so that they will not be processed again. It is obvious that this simplification still guaranties the completeness of the result, as at least one representative for each orbit of D-ER which is not in the excluded region will be included in at least one of the faces of the final output. The convenience of implementing this variation depends on the specific problem, and in particular on the ``size'' of the orbits and on how efficiently one can check the equivalence of various faces.

\paragraph{Simplicial cones:}
Another simple variations considers the possibility that some of the faces generated during the computation might be \textit{simplicial} cones. This is trivial to verify, since it only requires to check if the number of inequalities equals the dimension of the face. If a simplicial face is found, then instead of iterating the main subroutine for this face it is much more efficient to simply extract the D-ERs directly from the set of inequalities that specify the face.

\paragraph{Focusing on a lower-dimensional face:} In some situations one might be interested in finding only the D-ERs that belong to a lower-dimensional face $\face$ of the cone. For example, this is the case for 
the SAC, where as a result of \Cref{thm:bell-pair}, we can focus on the face $\face_*$ that we mentioned in the Introduction. In these situations, it would seem natural to first simplify the problem via a dimensional reduction to the subspace spanned by $\face$, and then apply the algorithm to the resulting set-up. However, as we discussed above when we commented on the post-processing of the ERs, the full structure of the poset is not ``completely visible'' in a lower-dimensional subspace. One practical alternative is to instead initialize the algorithm setting $\ds^{(0)} = \van(\face)$ instead of $\ds^{(0)} =\varnothing$.\footnote{\,Here we are assuming that $\face$ is a D-face. If it is not, then one should simply set $\ds^{(0)}=\fdc(\van(\face))$.}

\paragraph{Using different closure operators:}
In the main subroutine described in Algorithm~\ref{alg:main_subroutine}, we have used the closure operator $\fdc$. While this closure can be computed efficiently using a linear program \cite{Fukuda15}, for a large number of inequalities in a space with a large number of dimensions, this can still be a heavy computation. One alternative is to use $\ldc$ instead (see \S\ref{subsec:closures}), which can be computed much more easily. This is in fact what we used in the computations presented in \S\ref{sec:five} and \S\ref{sec:results}. Since this is a closure operator, the proof of completeness given in \Cref{thm:completeness} remains essentially unchanged. Each triplet now corresponds not necessarily to a face, but to a linear subspace such that $\ds$ is a down-set. We will call such a linear subpace a D-subspace. When there is a face $\face$ that spans this subspace, then $\face$ is a D-face and corresponds to one of the usual triplets. More generally, it can instead happen that the subspace only contains a proper lower-dimensional subspace spanned by a face, which is then not guaranteed to be a D-face. In practical terms however, the only difference in the algorithm is in the main subroutine, since other subroutines do not use any closure operator. Further, in the main subroutine, the only difference is that when we obtain a triplet corresponding to a 1-dimensional subspace, this subspace is now not guaranteed to be spanned by an ER, and we need to verify that it contains a ray that satisfies all the inequalities. Finally, when the algorithm is used not to find all D-ERs directly, but rather only to extract a simpler set of faces that includes all D-ERs, if we replace $\fdc$ with $\ldc$, we are now extracting a set of D-subspaces, although as we mentioned above they are not necessarily spanned by corresponding D-faces. The final result, however, is not affected by this difference, since the dimensional reduction of the remaining inequalities to this subspace will still specify a face of the cone, and even if this might not be a D-face, we will check during the post-processing which ERs we found are D-ERs.

\section{A simple example: $\N=5$}
\label{sec:five}

Having explained our algorithm in \S\ref{sec:algorithm} in full generality, we now return to the set-up described in \S\ref{sec:intro} and \S\ref{sec:mip} to exemplify the application of the algorithm to the SAC. Specifically, we will use it to compute the subset of $\erkc^5$ of ``genuine $5$-party'' ERs, i.e., those that are not lifts of element of $\erkc^{\N}$ for some $\N<5$. These ERs are well-known,\footnote{\,It was shown in \cite{He:2022bmi} that $\erkc^5=\erssa^5$, and $\erssa^5$ was computed in \cite{Hernandez-Cuenca:2019jpv}.} and our goal here is  simply to demonstrate how our algorithm works in a particular example. In the next section, we will use the algorithm to compute the much more involved $\erkc^6$ (with the main goal of finding $\erssa^6$), and there
 we will only present the result. 

To initialize the algorithm, we will use \Cref{thm:bell-pair} of \S\ref{subsec:mi-poset-and-kc} to implement a variation of the algorithm described above and focus on the lower-dimensional KC-face $\face_*$, which contains the KC-ERs not realized by Bell pairs. This means that to specify the initial triplet $\mathfrak{T}^{(0)}$, we can set 
\begin{equation}
    \ds^{(0)} =\van(\face_*)=\left\{\mi(\ell:\ell'),\;\text{for all $\ell,\ell' \in \psys{5}$}\right\} .
\end{equation}
Furthermore, since we are only focusing on genuine $5$-party KC-ERs, we can use the constraints described in \S\ref{subsec:constraints} and specify the excluded region $\regcone$ of the cone to be 
\begin{equation}
    \irr^{(0)} =\left\{\mi(\uJ: \uK),\; 5\leq |\uJ|+|\uK|\leq 6\right\} .
\end{equation}
The cardinalities of these sets are $\abs{\ds^{(0)}}=15$ and $\abs{\irr^{(0)}}=121$.
Lastly, the symmetry group $G$ is a subgroup of $\text{Sym}(270)$ (since $|\mgs|=301$, but the $31$ top elements of the MI-poset correspond to redundant inequalities) isomorphic to $\text{Sym}(6)$, which is the group of all permutations for the six parties in $\psys{5}$. Hence, we set $\stab{(0)}=G$.

Moreover, we will use two of the variations mentioned in \S\ref{subsec:variations} when implementing the basic algorithm presented in \S\ref{subsec:structure}.
The first variation is that we will use the closure operator $\ldc$ instead of $\fdc$ in the main subroutine. Each triplet will then correspond to a KC-subspace\footnote{\,As the reader might guess, a KC-subspace is precisely the D-subspace defined in \S\ref{sec:algorithm} when the poset is the MI-poset. Similarly, KC-faces and KC-ERs are respectively the D-faces and D-ERs defined above when the poset is the MI-poset.} rather than a KC-face. The reason for this choice is that computing $\ldc$ is typically much faster than computing $\fdc$, and empirical observation suggests that most instances of
KC-subspaces are in fact spanned by KC-faces. Thus, using $\fdc$ instead of $\ldc$ would not give a significant advantage, while dramatically slowing down the computation. We will continue to use the same strategy in the next section.

The second variation, which we only use here for the purpose of illustration, consists of checking if any of the faces generated during the computation is a simplicial cone. If so, we can simply extract the KC-ERs directly without having to continue iterating through the main subroutine of our algorithm. A schematic description of the various steps of the computation is given in \Cref{fig:N5-computation}, and we now describe a few of these steps in detail.

With these two additional variations in mind, and having specified the initial triplet $\mathfrak{T}^{(0)}$, we need to choose the criteria to stop the computation and to generate the triplets with the main subroutine. For $\N=5$, the computation is sufficiently simple that we stop only after processing all the triplets, therefore finding all the desired ERs without having to rely on other algorithms. This corresponds to setting the function $f$ in \Cref{cri:stop} to return TRUE only if the dimension of a KC-subspace is at most one. Meanwhile, a convenient choice for \Cref{cri:order} is the dimension of the KC-subspaces corresponding to the newly generated down-sets. In this way, the dimension of these subspaces will be increasing along the sequence, and subspaces with greater dimension will be associated to larger excluded regions.

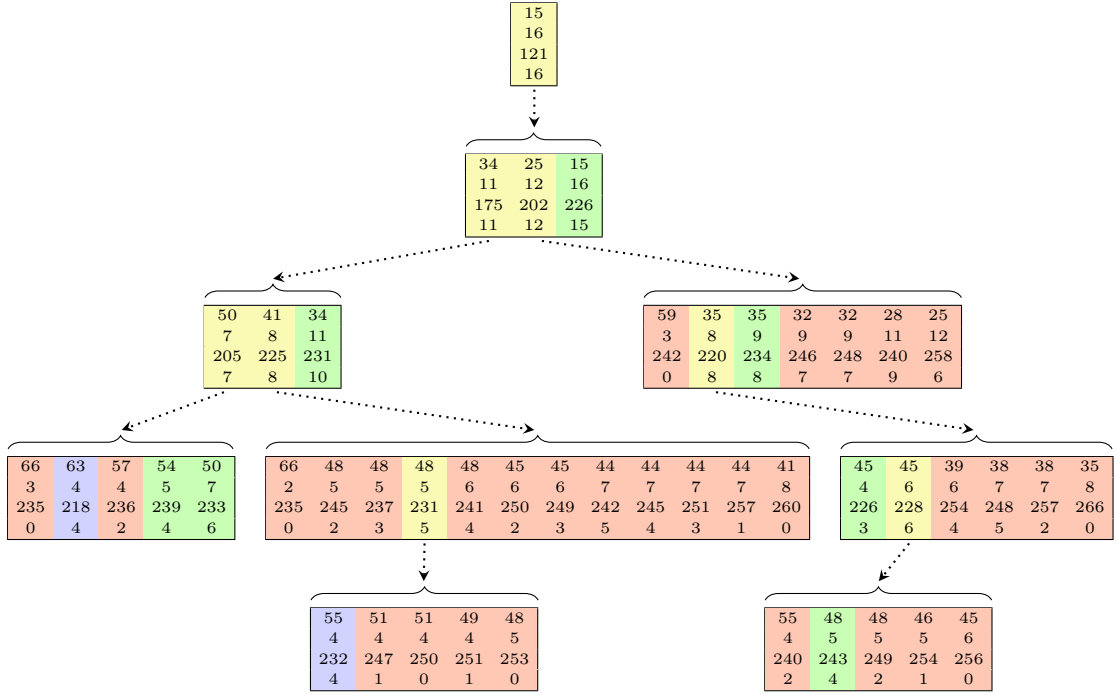
\begin{figure}[tb]
    \centering
    \begin{tikzpicture}

    \node[] at (-3.55,2) 
    {\tiny 
    \setlength{\tabcolsep}{3pt}
    \begin{tabular}{|y|}
    \hline
        15 \\
        16 \\
        121 \\
        16 \\
    \hline
    \end{tabular}
    };

    \draw[thick,dotted,-stealth] (-3.5,1.4) -- (-3.5,0.9);
    \draw [decorate,decoration={brace,amplitude=5pt}] (-4.4,0.65) -- (-2.62,0.65);

    \node[] at (-3.55,0) 
    {\tiny 
    \setlength{\tabcolsep}{3pt}
    \begin{tabular}{|yyg|}
    \hline
        34 & 25 & 15\\
        11 & 12 & 16\\
        175 & 202 & 226\\
        11 & 12 & 15\\
    \hline
    \end{tabular}
    };

    \draw[thick,dotted,-stealth] (-4.1,-0.6) -- (-6.95,-1.1);
    \draw[thick,dotted,-stealth] (-3.4,-0.6) -- (0,-1.1);
    \draw [decorate,decoration={brace,amplitude=5pt}] (-7.85,-1.35) -- (-6.055,-1.35);
    \draw [decorate,decoration={brace,amplitude=5pt}] (-2.05,-1.35) -- (2.14,-1.35);

    \node[] at (-7,-2) 
    {\tiny 
    \setlength{\tabcolsep}{3pt}
    \begin{tabular}{|yyg|}
    \hline
        50 & 41 & 34\\
        7 & 8 & 11\\
        205 & 225 & 231\\
        7 & 8 & 10\\
    \hline
    \end{tabular}
    };

    \node[] at (0,-2) 
    {\tiny 
    \setlength{\tabcolsep}{3pt}
    \begin{tabular}{|rygrrrr|}
    \hline
        59 & 35 & 35 & 32 & 32 & 28 & 25\\
        3 & 8 & 9 & 9 & 9 & 11 & 12\\
        242 & 220 & 234 & 246 & 248 & 240 & 258\\
        0 & 8 & 8 & 7 & 7 & 9 & 6\\
    \hline
    \end{tabular}
    };

    \draw[thick,dotted,-stealth] (-7.6,-2.6) -- (-8.95,-3.1);
    \draw[thick,dotted,-stealth] (-6.9,-2.6) -- (-3.5,-3.1);
    \draw[thick,dotted,-stealth] (-1.1,-2.6) -- (2.3,-3.1);
    \draw [decorate,decoration={brace,amplitude=5pt}] (-10.45,-3.35) -- (-7.48,-3.35);
    \draw [decorate,decoration={brace,amplitude=5pt}] (-7.04,-3.35) -- (0.13,-3.35);
    \draw [decorate,decoration={brace,amplitude=5pt}] (0.55,-3.35) -- (4.13,-3.35);

    \node[] at (-9,-4) 
    {\tiny
    \setlength{\tabcolsep}{3pt}
    \begin{tabular}{|rbrgg|}
    \hline
        66 & 63 & 57 & 54 & 50\\
        3 & 4 & 4 & 5 & 7\\
        235 & 218 & 236 & 239 & 233\\
        0 & 4 & 2 & 4 & 6\\
    \hline
    \end{tabular}
    };

    \node[] at (-3.5,-4) 
    {\tiny
    \setlength{\tabcolsep}{3pt}
    \begin{tabular}{|rrryrrrrrrrr|}
    \hline
        66 & 48 & 48 & 48 & 48 & 45 & 45 & 44 & 44 & 44 & 44 & 41\\
        2 & 5 & 5 & 5 & 6 & 6 & 6 & 7 & 7 & 7 & 7 & 8\\
        235 & 245 & 237 & 231 & 241 & 250 & 249 & 242 & 245 & 251 & 257 & 260\\
        0 & 2 & 3 & 5 & 4 & 2 & 3 & 5 & 4 & 3 & 1 & 0\\
    \hline
    \end{tabular}
    };

    \node[] at (2.3,-4) 
    {\tiny
    \setlength{\tabcolsep}{3pt}
    \begin{tabular}{|gyrrrr|}
    \hline
        45 & 45 & 39 & 38 & 38 & 35\\
        4 & 6 & 6 & 7 & 7 & 8\\
        226 & 228 & 254 & 248 & 257 & 266\\
        3 & 6 & 4 & 5 & 2 & 0\\
    \hline
    \end{tabular}
    };

    \draw[thick,dotted,-stealth] (-4.95,-4.6) -- (-4.96,-5.1);
    \draw[thick,dotted,-stealth] (1.45,-4.6) -- (1.05,-5.1);
    \draw [decorate,decoration={brace,amplitude=5pt}] (-6.45,-5.35) -- (-3.45,-5.35);
    \draw [decorate,decoration={brace,amplitude=5pt}] (-0.45,-5.35) -- (2.54,-5.35);

    \node[] at (-5,-6) 
    {\tiny
    \setlength{\tabcolsep}{3pt}
    \begin{tabular}{|brrrr|}
    \hline
        55 & 51 & 51 & 49 & 48\\
        4 & 4 & 4 & 4 & 5\\
        232 & 247 & 250 & 251 & 253\\
        4 & 1 & 0 & 1 & 0\\
    \hline
    \end{tabular}
    };

    \node[] at (1,-6) 
    {\tiny
    \setlength{\tabcolsep}{3pt}
    \begin{tabular}{|rgrrr|} 
    \hline
        55 & 48 & 48 & 46 & 45\\
        4 & 5 & 5 & 5 & 6\\
        240 & 243 & 249 & 254 & 256\\
        2 & 4 & 2 & 1 & 0\\
    \hline
    \end{tabular}
    };
    \end{tikzpicture}
    \caption{A schematic illustration of the various steps of the computation of $\erkc^5$, where each box corresponds to one iteration of the main subroutine. In each box, each column corresponds to one of the triplets in $\dot{\textgoth{T}}$ after the completion of line {\bf {\scriptsize 5}} in Algorithm~\ref{alg:main_subroutine}. For each triplet $\mathfrak{T}$, the four rows from top to bottom are as follows: the cardinality of $\ds$; the dimension of the KC-subspace specified by $\ds$; the cardinality of $\irr$; and the rank of   $\left.\mgs\setminus\!(\ds\cup\irr)\right|_{\ds}$, where the subscript $\ds$ indicates we are considering  the dimensional reduction of the inequalities to the KC-subspace specified by $\ds$, since we are working with KC-subspaces rather than D-faces. The various triplets are further distinguished by color depending on the following: whether they satisfy \Cref{cond:triplet_conditions} but the corresponding faces are not simplicial (yellow); whether they satisfy \Cref{cond:triplet_conditions} and the corresponding faces are simplicial (blue); whether \eqref{eq:1d-flat-condition} holds and they correspond to a D-ER (green); or whether \eqref{eq:non-pointed} is satisfied and they can be discarded (red). For each yellow triplet, the arrow indicates which box presents the data of the corresponding new iteration of the main subroutine. The box at the top shows the data for the initial triplet $\mathfrak{T}^{(0)}$, and the box immediately below it shows the data for the triplets $\mathfrak{T}^{(0,1)}$, $\mathfrak{T}^{(0,2)}$ and $\mathfrak{T}^{(0,3)}$, which are described in detail in the main text.
    }
    \label{fig:N5-computation}
\end{figure}

The first iteration of the main subroutine on the triplet $(\ds^{(0)},\irr^{(0)},\stab{(0)})$ begins by computing the maximal elements $\mathcal{M}^{(0)}$ of $\mgs\setminus(\ds^{(0)}\cup\irr^{(0)})$ and organizing them into $\stab{(0)}$-orbits. We obtain two orbits, which have the form 
\begin{align}
    & \mathcal{M}_1 = \left\{\mi(\uJ:\uK)\in\mgs,\; |\uJ|=2\;\; |\uK|=2 \right\}\nonumber\\
    & \mathcal{M}_2 = \left\{\mi(\uJ:\uK)\in\mgs,\; |\uJ|=3\;\; |\uK|=1 \right\} .
\end{align}
Choosing the representative $\mi(12:34)$ from the first orbit we get
\begin{align}
\begin{split}
    \ds_1 &=\ldc \left[\ds^{(0)}\cup \left\{\mi(12:34)\right\}\right] \\
    &= \left\{\! \text{ {\small$\mi(12:34),\mi(13:24),\mi(14:23),\mi(123:4),\mi(124:3),\mi(134:2),\mi(234:1)$}}\,\right\}_{\downarrow}\cup \ds^{(0)} ,
\end{split}
\end{align}
where $\{\cdots\}_{\downarrow}$ denotes the down-set generated by the elements in the argument. Hence $\abs{\ds_1}=34$.
Choosing instead the representative $\mi(123:4)$ from the second orbit gives
\begin{align}
\begin{split}
    \ds_2 &=\ldc \left[\ds^{(0)}\cup \left\{\mi(123:4)\right\}\right] \\
    &= \left\{\! \text{ {\small $\mi(123:4),\mi(14:2),\mi(14:3),\mi(24:1),\mi(24:3),\mi(34:1),\mi(34:2)$}}\,\right\}_{\downarrow} \cup\ds^{(0)}  ,
\end{split}
\end{align}
so $\abs{\ds_2}=25$. Since neither $\ds_1$ nor $\ds_2$ contain any element of $\irr^{(0)}$, they are not in the excluded region, and we have $\overline{\irr} =\irr^{(0)}$ (see line {\scriptsize {\bf 1}}  in Algorithm~\ref{alg:main_subroutine}). Furthermore, it is immediate to verify that these down-sets correspond to faces of dimension greater than one, and since they are clearly not related by a permutation of the parties, they are also the elements of $\ddot{\mathfrak{D}}$, the output of Algorithm~\ref{alg:permutations}. We can then apply to them (and the sets $\mathcal{M}_1$ and $\mathcal{M}_2$ in $\ddot{\mathfrak{M}}$) Algorithm~\ref{alg:triplets} to generate two new inequivalent triplets $(\ds^{(0,1)},\irr^{(0,1)},\stab{(0,1)})$ and $(\ds^{(0,2)},\irr^{(0,2)},\stab{(0,2)})$, with $\ds^{(0,1)} = \ds_1$ and $\ds^{(0,2)} = \ds_2$, and $\irr^{(0,1)}=\irr^{(0)}$ and $\irr^{(0,2)}=\irr^{(0)}\cup\mathcal{M}_1$. Furthermore, the groups $\stab{(0,1)}$ and $\stab{(0,2)}$ are the following subgroups of Sym(270):
\begin{align}
\begin{split}
    & \stab{(0,1)} \, \simeq \; \text{Sym}(4)_{_{1234}}\times\text{Sym}(2)_{_{50}} \\
    & \stab{(0,2)} \, \simeq \; \text{Sym}(3)_{_{123}}\times\text{Sym}(2)_{_{50}} ,
\end{split}
\end{align}
where the lower indices indicate the specific parties which are permuted by the various factors. Notice that as mentioned above, the ordering is chosen such that the dimension of the corresponding KC-subspace is increasing from $\ds_1$ to $\ds_2$. 

These two triplets are the input of Algorithm~\ref{alg:update}, which we can now use to ``extend'' the excluded region. There are 146 elements in $\mgs\setminus (\ds^{(0,1)}\cup\irr^{(0,1)})$, which can be organized into 11 different $\stab{(0,1)}$-orbits with the following representatives:
\begin{align}
\begin{split}
\label{eq:G1-orbits}
    & \text{\small $\mi(1:50)$, $\mi(15:0)$, $\mi(1:250)$,  $\mi(15:20)$, $\mi(12:50)$, $\mi(125:0)$} , \\
    & \text{\small $\mi(1:25)$, $\mi(12:5)$, $\mi(12:35)$, $\mi(123:5)$, $\mi(1:235)$} .
\end{split}
\end{align}
Similarly, in $\mgs\setminus (\ds^{(0,2)}\cup\irr^{(0,2)})$ there are 110 elements, and they can be organized into 21 different $\stab{(0,2)}$-orbits with the following representatives:
\begin{align}
\begin{split}
\label{eq:G2-orbits}
    & \text{\small $\mi(4:50)$, $\mi(45:0)$, $\mi(150:4)$,  $\mi(145:0)$, $\mi(1:250)$}, \\
    & \text{\small $\mi(1:234)$, $\mi(1:450)$, $\mi(1:45)$, $\mi(15:4)$, $\mi(1:50)$, $\mi(15:0)$, $\mi(14:5)$, $\mi(1:245)$}, \\
    & \text{\small $\mi(1:25)$, $\mi(125:4)$, $\mi(12:5)$, $\mi(125:0)$, $\mi(124:5)$, $\mi(1:235)$, $\mi(1:23)$, $\mi(123:5)$} .
\end{split}
\end{align}
One can verify that for any MI instance $\mi$ in the first row of \eqref{eq:G1-orbits}, or any MI instance $\mi'$ in the first row of \eqref{eq:G2-orbits}, we have
\begin{align}
    \begin{split}
    & \ldc\left[\ds^{(0,1)}\cup\,\left\{\mi\right\}\right]\cap\; \irr^{(0,1)}\neq\varnothing \\
    & \ldc\left[\ds^{(0,2)}\cup\,\left\{\mi'\right\}\right]\cap\; \irr^{(0,2)} \neq\varnothing.
    \end{split}
\end{align}
Therefore, Algorithm~\ref{alg:update} respectively adds to $\irr^{(0,1)}$ and $\irr^{(0,2)}$ all the elements belonging to the orbits of the MI instances in the first rows of \eqref{eq:G1-orbits} and \eqref{eq:G2-orbits}. The resulting triplets (with $\ds^{(0,1)}$ and $\ds^{(0,2)}$, and therefore $\stab{(0,1)}$ and $\stab{(0,2)}$, unchanged) automatically satisfy (iii) in \Cref{cond:triplet_conditions} by construction.

The next step of the main subroutine takes into account the possibility of relevant down-sets that include $\ds^{(0)}$ but no element of $\mathcal{M}^{(0)}$ (see line {\scriptsize {\bf 5}}  in Algorithm~\ref{alg:main_subroutine}), thereby generating a third triplet $(\ds^{(0,3)},\irr^{(0,3)},\stab{(0,3)})$, where 
\begin{align}
    \ds^{(0,3)} & =\ds^{(0)} \nonumber\\
    \irr^{(0,3)} & =\irr^{(0)} \cup\mathcal{M}^{(0)} \nonumber\\
    \stab{(0,3)} & =\stab{(0)}.
\end{align}

Having constructed these triplets, we now need to verify that they satisfy (iv) in \Cref{cond:triplet_conditions}. For the first two triplets, we leave it as an exercise for the reader to check that this is indeed the case. The last triplet, on the other hand, fails this condition, and we can immediately derive from it a ``candidate'' D-ER from the simultaneous saturation of all inequalities corresponding to the elements of $\mgs\setminus\irr^{(0,3)}$. This is a 1-dimensional KC-subspace generated by the entropy vector
\begin{equation}
    \vec{\ent}=\{1, 1, 1, 1, 1; 2, 2, 2, 2, 2, 2, 2, 2, 2, 2; 3, 3, 3, 3, 3, 3, 3, 3, 3, 3; 2, 2, 2, 2, 2; 1\} .
\end{equation}
Since $\vec{\ent}$ satisfies all instances of SA, it follows that it is an element of $\erkc^5$. A schematic summary of the further iterations of the main subroutine for the rest of the computation of $\erkc^5$ is shown in \Cref{fig:N5-computation}.

\section{Results for $\N=6$}
\label{sec:results}

Having exemplified the workings of our algorithm for the $\N=5$ case, we now present the results of the algorithm applied to finding the KC-ERs of SAC$_6$, primarily focusing on the ones that satisfy SSA.\footnote{\,Further details are presented in \cite{rays-notes}, in particular their adjacency relation and their characterization in terms of the $\beta$-sets of \cite{beta-sets}. Our heuristic procedure to obtain graph model realizations for most of them is described in \cite{graph-constructions}.} The algorithm produced 220 orbits of genuine $6$-party ERs in $\erkc^6$, and upon checking SSA, we find that 12 of these in fact violate at least one instance of SSA (cf.\ \S\ref{subsec:NonSSAERs}).\footnote{\,Unlike the $\N \leq 5$ computations, we find that for $\N=6$ there in fact even exist 1-dimensional KC-subspaces which are not extreme rays of SAC$_6$ because they violate at least one instance of SA. An example of such a subspace is that generated by {\tiny $\vec{\ent}=$\{4, 2, 3, 2, 2, 3; 6, 7, 6, 6, 7, 5, 4, 4, 5, 5, 5, 6, 4, 5, 5; 7, 8,
8, 9, 9, 9, 8, 8, 9, 9, 7, 7, 8, 6, 7, 5, 7, 8, 8, 7; 9, 9, 8, 10, 7,
9, 9, 8, 6, 9, 9, 10, 8, 11, 10; 7, 6, 6, 7, 6, 8; 4\}}.}

A natural first step towards classifying the orbits of $\erssa^6$ is to determine which ERs violate some \textit{holographic entropy inequality} (HEI). However, since the complete set of $\N=6$ HEIs is unknown, in principle this only specifies a subset of $\erssa^6$ which includes, but not necessarily coincides with, $\erh^6$. Using all the hitherto found HEIs (of which there are currently $1,\!877$, as 
detailed in \cite{Hernandez-Cuenca:2023iqh}), we find that there are 52 orbits which  violate at least one instance of these inequalities; these are discussed in \S\ref{subsec:HEIviolERs}.

The remaining 156 orbits are candidates for $\erh$, and in order to classify them as such, it suffices to find a holographic graph model which realizes the given entropy vector. In \S\ref{subsec:graphERs}, we summarize the result of this endeavor, presenting graphs for 150 ERs. This leaves only 6 unclassified orbits, which we dub ``mystery'' ERs, and we further comment on them in \S\ref{subsec:mysteryERs}.

\begin{table}[htbp] 
\begin{center}
\includegraphics[width=\textwidth]{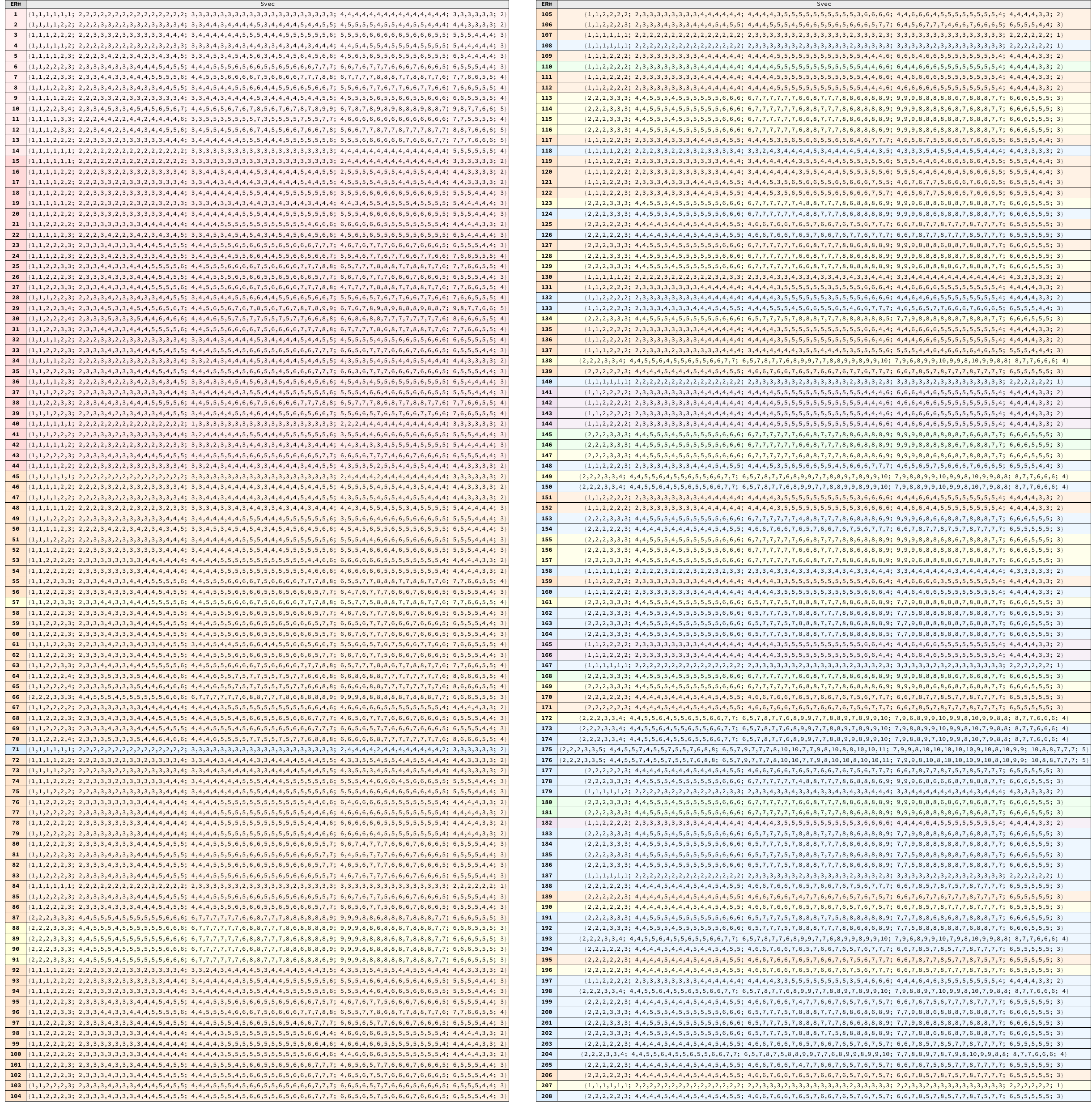}
\caption{
    The entropy vector of the representative instance $\vec{\ent}_s$, for all orbits of ERs in $\erssa^6$, split into two columns.
    The {\bf rows} list the ERs, labeled by the index $s$, and color-coded by type: Warm colors correspond to ERs realized by holographic graph models, further subdivided into star graphs (pink), other simple tree graphs (red), other ERs realizable by graphs which were previously known (orange), and new ERs which are realizable by graphs (yellow). Cold colors correspond to HEI-violating ERs, further subdivided into those ERs that violate MMI (blue) and those ERs that only violate HEI for more than three parties (purple). Finally, the color green corresponds to mystery ERs. The components of these vectors are ordered as in previous works, first by cardinality of $\J$, and then, within each group, by lexicographical order. For instance, the entropies of individual parties comprise the first six entries, as well as the last one (corresponding to the purifier).
}
\label{tab:Svectors}
\end{center}
\end{table} 

In \Cref{tab:Svectors}, we present the entropy vectors of a representative instance from every orbit in $\erssa^6$ (this data can also be found in the supplementary materials), where we did not include the lifts of the elements of $\erssa^6$ for $\N< 6$.
Following the usual convention, the components of each vector are arranged by the cardinality of the subsystems $\J$ (separated by semicolon for ease of identification), and within each group they are arranged in lexicographic order.
For convenience, we also color-code each ER by its type: Warm colors (pink, red, orange, yellow) represent holographic ERs (those definitely in $\erh$), cold colors (blue, violet) represent ERs that violate at least one HEI (those definitely in $\erssa \setminus \erh$), and green represents the remaining mystery ERs.  This color scheme is further subdivided as follows.  Pink ERs are realized by star graphs, red ones by other simple\footnote{\,By \emph{simple} tree graphs, we mean graphs which have no cycles and have one boundary vertex per party.
}
tree graphs, and orange and yellow ones by more complicated graphs. The orange ones were known previously \cite{hecdata}, while the yellow ones are new, illustrating the power of our construction. Notice that since all  ERs of the SAC which can be realized by graph models are simultaneously ERs of the HEC, the yellow ERs yield previously unknown ERs of the $\N=6$ HEC. Finally, among the ERs that violate some HEI, we distinguish the ones which violate MMI by blue and the ones which only violate some higher $\N=5,6$ HEI by purple.

Labeling each orbit in $\erssa^6$ by an index $s\in[208]$,
the coarse breakdown of holographic versus non-holographic ER orbits is given by the following three categories.
\begin{itemize} 
\item 
The {\bf 52}  ``non-holographic'' ERs that violate some HEI, and are therefore in $\erssa \setminus\erh$, are generated by the entropy vectors $\vec{\ent}_s$, with $s$ in
\{{\footnotesize 71, 108, 118, 124, 126, 132, 133, 140, 141, 142, 143, 144, 148, 150, 
153, 154, 158, 160, 162, 163, 164, 165, 166, 167, 173, 174, 175, 176, 
177, 178, 179, 182, 183, 184, 185, 186, 187, 188, 191, 192, 193, 194, 
197, 198, 199, 200, 201, 202, 203, 204, 205, 208}\}.
We discuss these further in \S\ref{subsec:HEIviolERs}.
\item 
The {\bf 150} ERs that can be realized by holographic graph models, and are therefore in $\erh$, are generated by the entropy vectors $\vec{\ent}_s$, with $s$ in 
\{{\footnotesize 1, 2, 3, 4, 5, 6, 7, 8, 9, 10, 11, 12, 13, 14, 15, 16, 17, 18, 19, 
20, 21, 22, 23, 24, 25, 26, 27, 28, 29, 30, 31, 32, 33, 34, 35, 36, 
37, 38, 39, 40, 41, 42, 43, 44, 45, 46, 47, 48, 49, 50, 51, 52, 53, 
54, 55, 56, 57, 58, 59, 60, 61, 62, 63, 64, 65, 66, 67, 68, 69, 70, 
72, 73, 74, 75, 76, 77, 78, 79, 80, 81, 82, 83, 84, 85, 86, 87, 88, 
89, 90, 91, 92, 93, 94, 95, 96, 97, 98, 99, 100, 101, 102, 103, 104, 
105, 106, 107, 109, 111, 112, 113, 114, 115, 116, 117, 119, 120, 121, 
122, 123, 125, 127, 128, 129, 130, 131, 134, 135, 136, 137, 138, 139, 
147, 149, 151, 152, 155, 156, 157, 159, 161, 169, 170, 171, 172, 189, 
190, 195, 196, 206, 207}\}.
We discuss these further in \S\ref{subsec:graphERs}.
\item 
Finally, the remaining {\bf 6}  ``mystery'' ERs that do not violate any known HEI, but thus far do not have a holographic graph model construction, are generated by the entropy vectors $\vec{\ent}_s$, with $s$ in
\{{\footnotesize 110, 145, 146, 168, 180, 181}\}.
We remark on these in \S\ref{subsec:mysteryERs}.
\end{itemize} 

\begin{table}[htbp] 
\begin{center}
\includegraphics[width=\textwidth]{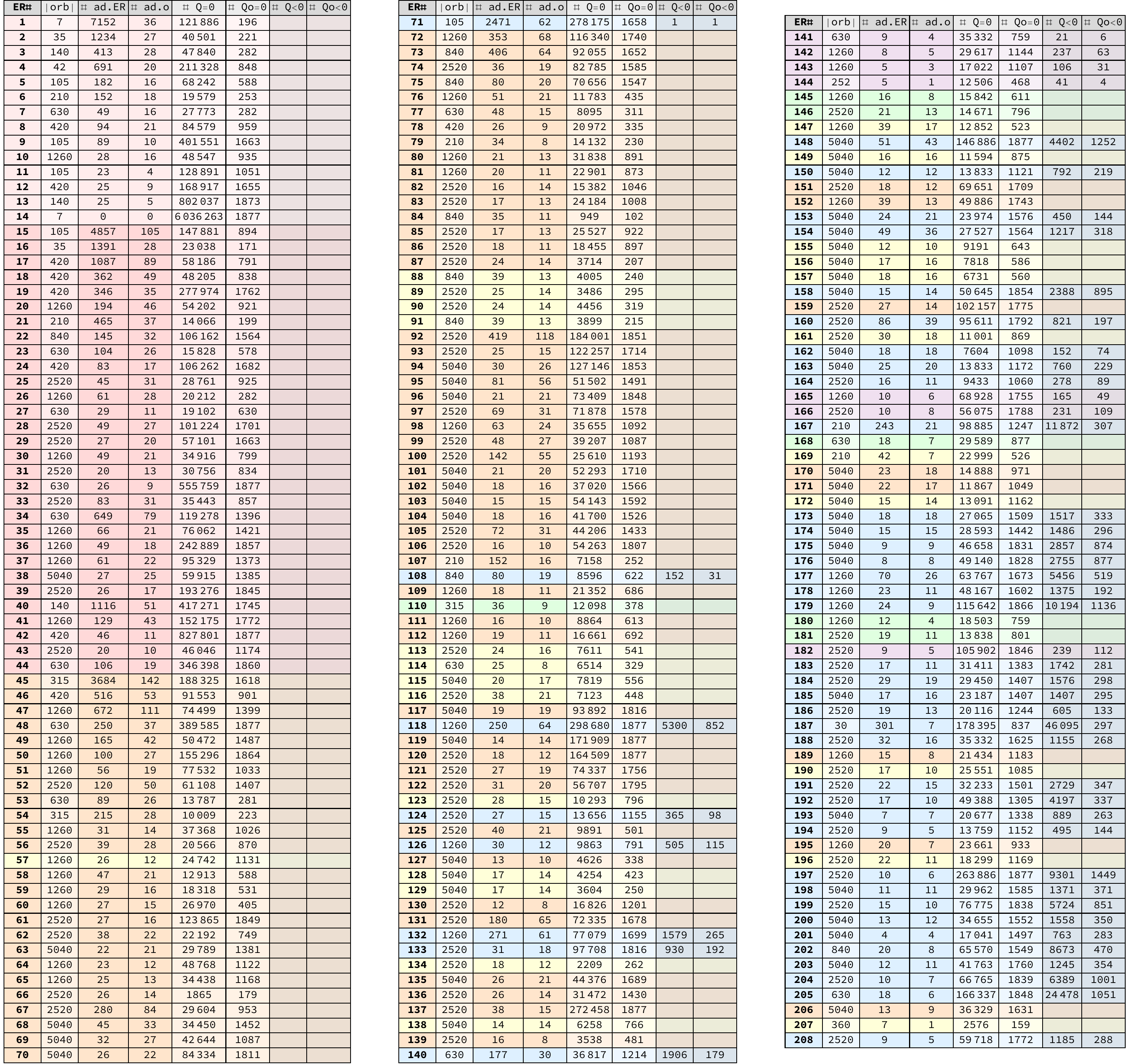}
\caption{
    Fundamental data about all orbits of ERs in $\erssa^6$, split into three columns,
    with the same color-coding as in \Cref{tab:Svectors}.
    The {\bf columns} specify several attributes, visually separated by shading: the size of the orbit (lighter), the number of saturating HEIs and HEI orbits (normal), and the number of violating HEIs and HEI orbits (darker).
}
\label{tab:sumtab}
\end{center}
\end{table} 

Before focusing on these categories separately in the subsections below, in \Cref{tab:sumtab} we present several further attributes of the ERs in $\erssa^6$.  Maintaining the same color-coding as in \Cref{tab:Svectors}, we list for each orbit: 
\begin{itemize} 
\item The number of ERs in the orbit (``$|\text{orb}|$''), which ranges from 7 (for two orbits realizable by star graphs) to $7!=5,\!040$ (for 45 orbits).
\item The total number of individual HEIs (from the $1,\!877$ orbits listed in \cite{Hernandez-Cuenca:2023iqh}, including the non-redundant instances of SA) which are saturated by the given ER (``\#Q=0''), and the number of orbits of HEIs containing a saturating instance  (``\#Qo=0'').  
The former are different for almost all orbits of ERs, and range over three orders of magnitude.
For the latter, the smallest number of HEI orbits containing a saturated instance is 102, while for 10 ERs each of the $1,\!877$ HEI orbits contains a saturated instance.
\item The total number of individual HEIs which are violated by the given ER (``\#Q$<0$''), and the number of HEI orbits containing a violated instance  (``\#Qo$<0$'').  These are of course non-zero only for the non-holographic ERs, so for the holographic and mystery ones we leave the entry blank to avoid clutter.  The largest number of violated instances is $46,\!095$, whereas the largest number of HEI orbits containing violated instances is $1,\!449$.  In \S\ref{subsec:HEIviolERs}, we will provide further details by distinguishing the types of HEIs.
\end{itemize} 

\subsection{The approximation of SSA by KC is not exact}
\label{subsec:NonSSAERs}

\begin{table}[htbp] 
\begin{center}
\includegraphics[width=\textwidth]{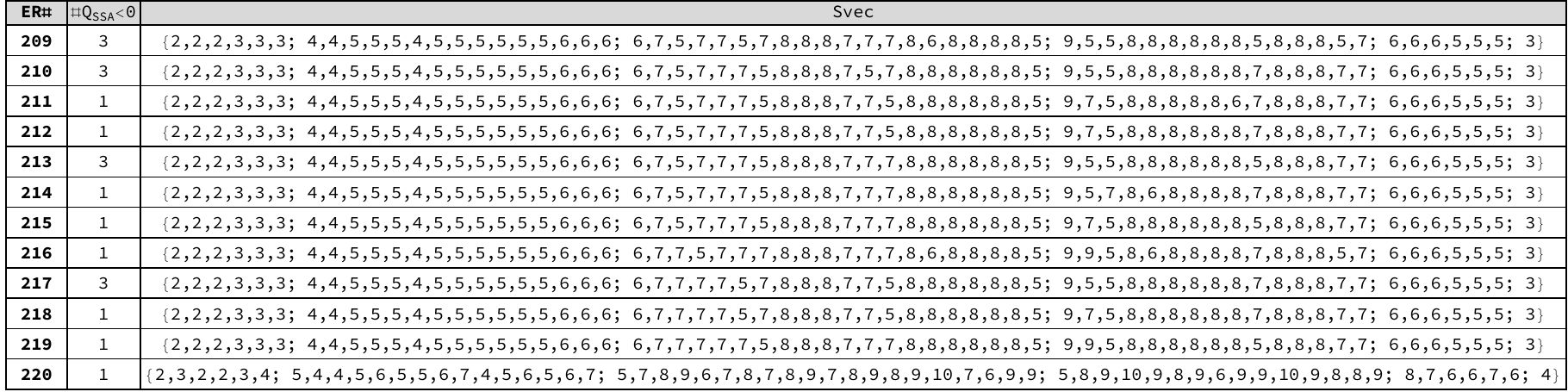}
\caption{
    The entropy vectors for a choice of representatives for all orbits of ERs in $\erkc^6 \setminus \erssa^6$.
    The second column indicates how many instances of SSA are violated by each given ER.
}
\label{tab:SvectorsSSAviol}
\end{center}
\end{table} 

As anticipated above, some ERs in $\erkc^6$ violate SSA, which means that in contrast to the situation for smaller values of $\N$ studied in \cite{He:2022bmi}, for $\N \ge 6$ there is in fact a gap between $\erkc$ and $\erssa$. This means that the approximation of SSA by KC is not exact even for extreme rays. In particular, there are 12 orbits of ERs in $\erkc^6 \setminus \erssa^6$, and in \Cref{tab:SvectorsSSAviol}, we list the entropy vector of a representative for each orbit (using the same conventions as in \Cref{tab:Svectors}). The table also shows the number of SSA instances which are violated by the given ER. Notice that there are very few instances that are violated: In four cases there are three, and in the remaining cases there is only a single one. While these ERs are not particularly interesting for any application, a more careful exploration of the combinatorial structure of their corresponding down-sets in the MI-poset might be useful to refine KC to a stronger condition that better approximates SSA and possibly improves the efficiency of the algorithm.

\begin{table}
\begin{center}
\includegraphics[width=4in]{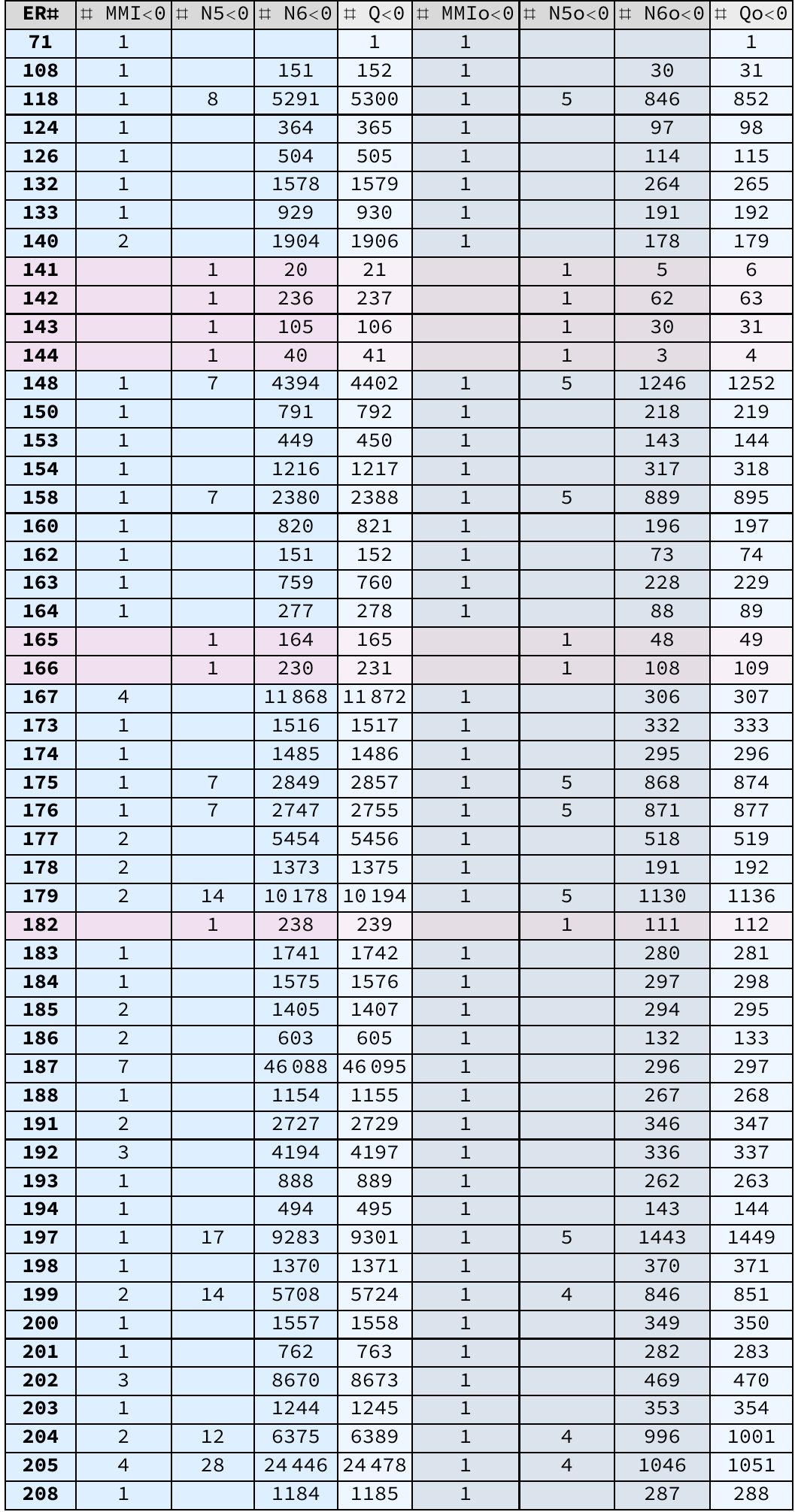}
\caption{
    A summary of the orbits of ERs in $\erssa^6$ that violate at least one known HEI. The {\bf rows} list the ER label $s$, with the same color-coding as in \Cref{tab:Svectors}.
    The {\bf columns} specify the total number of HEIs which are violated by the given ER, broken down into MMI, $\N=5$ HEIs, $\N=6$ HEIs (normal color) and the total number of HEI instances (lighter color), as well as the corresponding number of orbits admitting a violated instance (darker color) and the total (lighter).
}
\label{tab:nonhol}
\end{center}
\end{table} 

\subsection{Extreme rays that violate holographic entropy inequalities}
\label{subsec:HEIviolERs}

As we mentioned above, 52 of the 208 orbits of ERs in $\erssa^6$ violate some known HEI, and it is natural to further categorize these inequalities into three families: \textit{monogamy of mutual information} (MMI) \cite{Hayden:2011ag} (including its lifts), $\N=5$ HEIs \cite{Bao:2015bfa,Cuenca:2019uzx}, and $\N=6$ HEIs \cite{Hernandez-Cuenca:2023iqh} (the latter list is incomplete, so for this case the number of violations is merely a lower bound).  Regarding the instances of MMI, up to permutations and purifications there are only three distinct possibilities for the cardinalities of the arguments of the \textit{tripartite information}\footnote{\,The tripartite information is defined as $\mi_3(\uI:\uJ:\uK)=\ent_{\uI}+\ent_{\uJ}+\ent_{\uK}-\ent_{\uI\uJ}-\ent_{\uI\uK}-\ent_{\uJ\uK}+\ent_{\uI\uJ\uK}$, and MMI is the statement that $\mi_3(\uI:\uJ:\uK)\le0$
.}: (1:1:1), (1:1:2), and (1:2:2). As it turns out, only the last type can be violated.

In \Cref{tab:nonhol}, we list the number of violations by family (MMI / $\N=5$ / $\N=6$), both for individual instances and for their orbits. For convenience, we also present the total number of violations (which has already appeared in \Cref{tab:sumtab}). There are 45 MMI-violating ER orbits, out of which only \#71 does not violate any HEI for larger $\N$. All the remaining 51 orbits instead violate some $\N=6$ HEI, and 17 of them also violate some $\N=5$ HEI. On the other hand, not a single orbit violates only HEI for $\N=6$, and only seven orbits do not violate any instance of MMI. This seems to suggest that at a given $\N$, it is sufficient to know the HEI for some values of $\N'\leq\N$ to determine which elements of $\erssa$ are in $\erh$, and it would be interesting to explore this behavior more deeply.

Although we do not have good control over $\erq$ beyond the nesting indicated by \eqref{eq:N6_R_nesting}, we have previously constructed hypergraph realizations for some of these ERs. In particular, \#118 is given by (a permutation of) the hypergraph in \cite[Fig.1]{He:2023cco} and \#71 is given by the hypergraph in \cite[Fig.1]{He:2023aif}. We stress however that we have not attempted here to realize other ERs by hypergraphs, and it is entirely possible that they are in fact all realizable. We have checked the \textit{hypergraph inequality} proven in \cite{Bao:2020mqq}, as well all instances of Ingleton inequality \cite{Ingleton} (the main inequality characterizing the entropies of stabilizer states for four parties \cite{Linden:2013kal}), and found no violation.

\subsection{Extreme rays realizable by holographic graph models}
\label{subsec:graphERs}

\begin{table}[htbp] 
\begin{center}
\includegraphics[width=\textwidth]{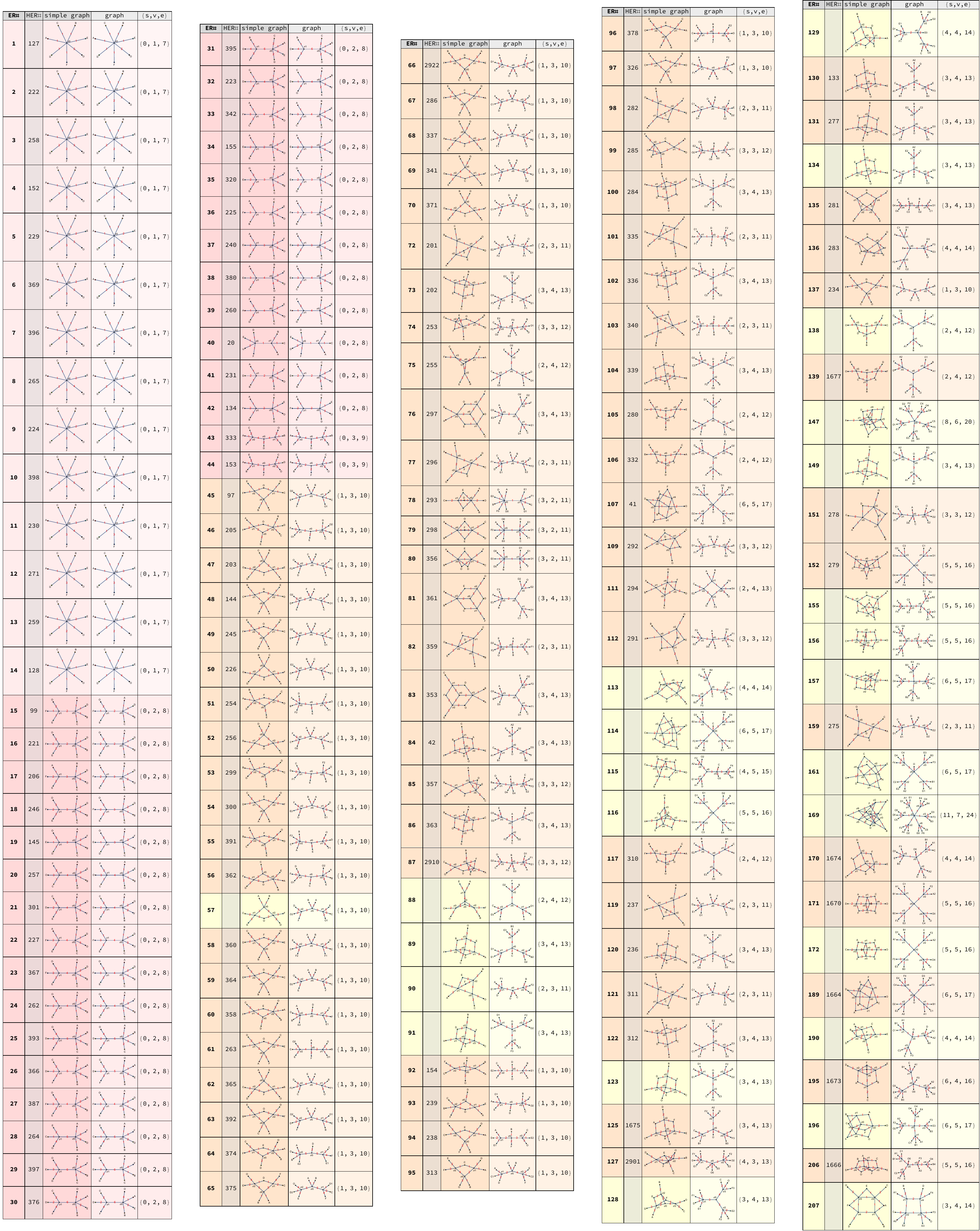}
\caption{
    The orbits of ERs in $\erssa^6$ for which we have constructed a holographic graph model. The {\bf rows} list the ER label $s$, with color-coding as in \Cref{tab:Svectors}. 
    The {\bf columns} specify: the corresponding label from \cite{hecdata}  for ERs which were previously known (darker),
    two equivalent graphs (the first with seven boundary vertices, the second with the boundary vertices split so as to avoid cycles containing boundary vertices), and the corresponding graph attributes \{s,v,e\} (specifying the \#  of splits, \# of bulk vertices, and \# of edges). To avoid confusion with the values of the edge weights, the parties $0,1,\ldots,6$ have been relabeled in the various graphs as $O,A,\ldots,F$, respectively.
}
\label{tab:graphERs}
\end{center}
\end{table} 

We now turn to the largest category, the 150 orbits of ERs which we explicitly realize by holographic graph models and therefore lie 
in $\erh$.
Out of these orbits, 125 were already known to be holographic (although we found a simpler graph realization for 81 of them), while the remaining 25 are new. We present these results in \Cref{tab:graphERs}, categorized by the same color-coding scheme as in the previous tables.  For the former set, we indicate the corresponding label from \cite{hecdata}.  We present two equivalent graphs for each ER: a simple version with one boundary vertex per party (these graphs therefore contain cycles for all but the 44 graphs corresponding to simple trees), and a ``tree-ified'' version where the boundary vertices are ``fine-grained''  so as to remove all cycles which include boundary vertices.  In agreement with the strongest form of the conjecture in \cite{Hernandez-Cuenca:2022pst}, most\footnote{\,For two ERs (\#111 and \#207), their corresponding graphs still each contains a ``bulk cycle'', and it would be interesting to see if we can alternatively realize them by tree graphs. We leave this for future exploration.
}
of these graphs are tree graphs, though they may not be simple tree graphs. However, this means that they would represent simple tree graphs in a system with $\N'>\N$ parties (for some appropriate $\N'$). In fact, for such  $\N'$, one can immediately verify that upon fine-graining, all these graphs realize ERs of SAC$_{\N'}$.
To quantify the necessary amount of fine-graining (or equivalently, the number of individual ``splits'' of boundary vertices), we indicate the graph complexity in the last column of  \Cref{tab:graphERs}.  In particular, we list three attributes, denoted as \{s,v,e\}, which respectively specify the number  of splits, the number of bulk vertices, and the total number of edges in the graph. Notice that many of the tree-ified graph models share the same underlying topology. 
Furthermore, note that the edge weights have the intriguing property that the sum of the weights of all edges incident on any given vertex is even.

Finally, we remark that some of the ERs with more complicated graphs could potentially have an even simpler graph realization. However, \cite{beta-sets}  presents a necessary condition for the realizability of an arbitrary entropy vector by a simple tree graph model.
We checked this condition for all 220 orbits in $\erkc^6$ (including the ones that violate SSA), and the only orbits that satisfy this condition are precisely \#1 to 
\#44, which are the ones we realized by simple trees. This is perhaps an indication that the condition found in \cite{beta-sets} might in fact
be sufficient as well.

\subsection{Mystery extreme rays}
\label{subsec:mysteryERs}

Finally, the six remaining ER orbits, for which we have found neither graph realizations nor any instances of violated HEIs, remain as yet unclassified. It is possible that the complexity of graphs realizing them may be much higher, but it is also possible that they may violate an $\N=6$ HEI not found in \cite{Hernandez-Cuenca:2023iqh}. Apart from the absence of both graph realizations and HEI violations, the attributes of these mystery ER orbits do not manifest any obvious differences from either the holographic ERs or the HEI-violating ERs. As we mentioned above, all the orbits in $\erssa^6$ which we were able to classify as non-holographic ER orbits violate at least one HEI for $\N\leq 5$. These six mystery ER orbits might be an exception, but based on this observation it is tempting to speculate that they are indeed holographic.

\section{Discussion}
\label{sec:discussion}

One of the main goals of this work was to present an algorithm for computing the extreme rays of a polyhedral cone when the defining inequalities have a partial order and we are only interested in the ERs corresponding to down-sets. Although we have presented this algorithm in general terms, our primary motivation was to compute the set $\erssa^6$, which is the set of ERs of the SAC$_6$ satisfying SSA, in order to test certain observations and conjectures from previous works (particularly 
\cite{Hernandez-Cuenca:2019jpv,Hernandez-Cuenca:2022pst,He:2022bmi}). Indeed, the set $\erssa^6$ has turned out to be much richer than the analogous sets for smaller $\N$. Specifically, $\erssa^6$ contains 208 orbits of genuine 6-party rays, whereas $\erssa^5$ contains only six orbits of genuine 5-party ones. We have already utilized part of the data from this computation to derive new results that have appeared in earlier publications \cite{He:2023aif,He:2023cco}. Here, we highlight some directions for future work, which we organize into two main topics: (1) considerations regarding the algorithm and potential new applications, and (2) comments about open questions specific to the elements of $\erssa^6$.

\paragraph{Generalizations and improvements of the algorithm:}

The algorithm we have presented is tailored to search for KC-ERs, but there are also reasons to compute higher-dimensional KC-faces. For example, as mentioned in the Introduction, the set of all SSA-compatible faces of the SAC$_\N$ forms a lattice where the meet operation is simply set intersection. To advance towards solving the quantum marginal independence problem, we can focus on the ``building blocks'' of this lattice, namely its meet-irreducible elements. It was shown in \cite{He:2022bmi} that even for $\N=4$, some of these elements are higher-dimensional faces. Furthermore, according to the construction in \cite{He:2023aif}, at least for $\N=4$, these higher-dimensional faces contain projections of SSA-compatible KC-ERs for some $\N' > 4$, corresponding to coarse-grainings of the parties. Computing higher-dimensional SSA-compatible faces and ERs for larger $\N$ may reveal whether this phenomenon at $\N=4$ is a coincidence, or whether it is indicative of a more general structural property of the SSA-compatible (and perhaps even realizable) faces of the SAC. While ERs are certainly meet-irreducible, higher-dimensional meet-irreducible faces might be harder to find, necessitating a generalization of our algorithm to search for higher-dimensional faces, possibly with further requirements.

Another interesting direction to explore with such a generalization of our algorithm is to search for KC-ERs that have additional intriguing properties. For example, all SSA-compatible KC-ERs for $\N \leq 6$ satisfy Ingleton's inequality \cite{Ingleton} (which is known to be satisfied by all stabilizer states \cite{Linden:2013kal}). This might lead one to speculate that they all lie inside the stabilizer entropy cone, which is unknown for $\N \geq 5$. One could then adapt our algorithm to search, at increasing values of $\N$, for SSA-compatible ERs that, for example, violate Ingleton's inequality.

Finally, since these questions are specific to faces of the SAC, it would be worthwhile to explore whether the algorithm could be reformulated using the $\beta$-sets technology introduced in \cite{beta-sets}. This formulation takes advantage of additional structural properties of SSA-compatible faces and, if incorporated into the algorithm, might help restrict the search space and improve efficiency.

\paragraph{Open questions about the SSA-compatible extreme rays of the 6-party SAC:}

The weak form of the conjecture proposed in \cite{Hernandez-Cuenca:2022pst} states that for any given $\N$, all extreme rays of the $\N$-party holographic entropy cone can be obtained from subsystem coarse-grainings of extreme rays in $\erh^{\N'}$ for some $\N' \geq \N$. A stronger form of this conjecture asserts that every extreme ray of the HEC can be realized by a graph model with tree topology (although not necessarily \textit{simple}; that is, not all boundary vertices are required to be labeled by distinct parties).\footnote{\,We refer the reader to \cite{Hernandez-Cuenca:2022pst} for the proof that the strong form implies the weak one.} One of the motivations of this work was to generate data that could be used to corroborate these conjectures.

Out of the 208 orbits of extreme rays (ERs) in $\erssa^6$, 52 violate at least one HEI and can therefore be disregarded for this purpose. For the remaining 156 orbits, we were able to construct tree realizations for 148 of them, supporting the strong form of the conjecture. Future attempts to find a counterexample should focus on the remaining eight orbits. Of these, six remain unresolved (see \S\ref{subsec:mysteryERs}), as we were unable to determine whether they can be realized by any graph model. It might be worthwhile to explore whether they can be used to propose new HEIs, or if a more focused search could provide a graph realization. The other two orbits (\#111 and \#207) are definitively holographic, as we did find a graph realization for them, but we were unable to find a tree graph realization. If the strong form of the conjecture is indeed false, perhaps one of these ERs is a counterexample. Proving this, however, would require showing that such a tree does not exist for any possible fine-graining to $\N' \geq \N$ parties, which might necessitate the development of new techniques.

Finally, as previously mentioned, we have not attempted to construct hypergraph models to realize the ERs in \S\ref{subsec:HEIviolERs} and \S\ref{subsec:mysteryERs}. However, this might be a useful approach to explore whether all these ERs can be realized by stabilizer states and if any of the inclusions in $\erhg \subseteq \ersta \subseteq \erq$ are strict. Of course, obtaining a definitive answer to this question might require more powerful techniques to explore the realizability of the elements of $\erssa$ beyond stabilizer states, and to determine whether even the inclusion $\erq \subseteq \erssa$ is strict. This problem, however, is far beyond the scope of this work.

\acknowledgments

We would like to especially thank David Simmons-Duffin for generously providing us supercomputing time in the early iterations of the algorithm. We would also like to thank Sergio Hern\'andez-Cuenca and Sridip Pal for useful discussions.  

T.H. is supported by the Heising-Simons Foundation ``Observational Signatures of Quantum Gravity'' collaboration grant 2021 -2817, the U.S. Department of Energy, Office of Science, Office of High Energy Physics, under Award No. DE-SC0011632, and the Walter Burke Institute for Theoretical Physics.  V.H. has been supported in part by the U.S. Department of Energy grant DE-SC0009999 and by funds from the University of California.   M.R. acknowledges support from UK Research and Innovation (UKRI) under the UK government’s Horizon Europe guarantee (EP/Y00468X/1), and by funds from the University of California during the early stages of this work.

V.H. acknowledges the hospitality of the Kavli Institute for Theoretical Physics (KITP), the Aspen Center for Physics, and the Tsinghua Southeast Asia Center, where part of this work was done. V.H. and M.R. would like to thank the Yukawa Institute for Theoretical Physics for hospitality during the program ``Quantum Information, Quantum Matter and Quantum Gravity'' during the early stages of this work.

All data is available within the paper, as an ancillary file in the arXiv, and on GitHub \cite{github}. The code used to generate and analyze this data is a combination of a Wolfram Mathematica implementation of the algorithm outlined in this work, and the freely available software Normaliz \cite{Normaliz:301}.

For the purpose of open access, the authors have applied a Creative Commons Attribution (CC BY) licence to any Author Accepted Manuscript version arising from this submission.

\appendix

\section{Comments on the relation between closure operators and symmetries}
\label{appendix}

Given a set $\mgs$, consider the symmetric group on $\mgs$, denoted by $\text{Sym}(|\mgs|)$. This is the group of all possible permutations of $\mgs$. The action of $\text{Sym}(|\mgs|)$ on $\mgs$ can naturally be lifted to the power-set $2^{\mgs}$ of $\mgs$ by defining the element-wise action, i.e., given an arbitrary subset $X=\{\E_1,\E_2,\ldots\}\subseteq\mgs$, we define
\begin{equation}
    g(X)=\{g(\E_1),g(\E_2),\ldots\}.
\end{equation}

Consider now an arbitrary closure operator \textit{on} $\mgs$ denoted $\cl:2^{\mgs}\rightarrow 2^{\mgs}$, which is a map from the power-set of $\mgs$ to itself with the defining properties outlined in \S\ref{subsec:closures}. The power-set $2^{\mgs}$ admits a lattice structure where the meet is given by the intersection and the join by the union. Furthermore, the subset of the \textit{set} $2^{\mgs}$ whose elements are closed with respect to $\cl$ (the closed subsets of $\mgs$) also admits a lattice structure, which we denote as $\lat{\cl}$. However, this is \textit{not a sublattice} of the power-set lattice $2^{\mgs}$, since while the meet is still given by the intersection, the join is in general not given by the union. Specifically, we have in general
\begin{align}
    & \cl(X \cap Y) = \cl(X) \cap \cl(Y), \nonumber\\
    & \cl(X \cup Y) = \cl(\cl(X) \cup \cl(Y)) \neq \cl(X) \cup \cl(Y).
\end{align}
Notice that the second expression clarifies why it is typically harder to understand the join (see for example \cite{He:2022bmi}), since given two closed sets $X$ and $Y$, computing their join still involves computing the closure. In other words, if we view these lattices as algebraic structures, $\cl$ is \textit{not} in general a lattice homomorphism between $2^{\mgs}$ and $\lat{\cl}$, because it does not respect the join.

Finally, notice that since $\mgs$ is a finite set, $2^{\mgs}$ and (the set) $\lat{\cl}$ are also finite, implying that the two corresponding lattices are \textit{complete} \cite{Davey_Priestley_2002}. This means that not only does the meet and join exist for any \textit{two} elements (which is the defining property of a lattice), but that they also exist for \textit{arbitrary collections} of elements.

\paragraph{When do the group actions and closure commute?} Having specified a group action on $2^{\mgs}$ and introduced a closure operator on $\mgs$, we now want to explore the relation between these two maps, and in particular, we want to know under what conditions they commute. Specifically, having fixed the choice of a closure operator $\cl$, we want to know which $g\in\text{Sym}(|\mgs|)$ satisfies
\begin{equation}
\label{eq:commute}
    \cl(g(X)) = g(\cl(X)), \qquad \forall\, X\in 2^{\mgs} .
\end{equation}
Notice that in the expression above, the left-hand side is a closed set, while in general the right-hand side need not be if we do not impose any restriction on the choice of $\cl$ and $g$.\footnote{\,For example, in \Cref{fig:3d_cone} if $g=(23)$ (in cycles notation) and $X=\{1,2\}$, then
\begin{equation*}
    \fc(g(\{1,2\}))=\fc(\{1,3\})=\{1,2,3,4\}\neq g(\fc(\{1,2\})) = g(\{1,2\}) = \{1,3\}.
\end{equation*}
}
It follows that a basic necessary condition for \eqref{eq:commute} to hold is that $g$ must be chosen such that the \textit{set} $\lat{\cl}$ is closed with respect to $g$, i.e., $g$ maps closed sets to closed sets. We can then consider the largest subgroup $G$ of $\text{Sym}(|\mgs|)$ such that each $g\in G$ has this property.

To proceed further, it is useful to consider the following general property of a closure operator. For any $X\subseteq\mgs$, the closure of $X$ is equal to the intersection of all closed sets that contain $X$, namely
\begin{equation}
\label{eq:closure_from_int}
    \cl(X)=\bigcap\; \{Y,\, Y\in \lat{\cl}\; \text{and}\; X\subseteq Y\}.
\end{equation}
Furthermore, the action of $G$ on $2^{\mgs}$ is a lattice homomorphism, i.e., for all $g$ and $X,Y\in 2^{\mgs}$, we have
\begin{align}
    & g(X\cap Y) = g(X) \cap g(Y)    \nonumber\\
    & g(X\cup Y) = g(X) \cup g(Y),
\end{align}
and since $2^{\mgs}$ is a complete lattice, analogous relations hold for arbitrary collections of elements of $2^{\mgs}$. From this property and \eqref{eq:closure_from_int}, it follows that
\begin{equation}
    g(\cl(X))=\bigcap\; \{g(Y),\, Y\in \lat{\cl}\; \text{and}\; X\subseteq Y\}.
\end{equation}
It is convenient to rewrite the family of sets which appears on the right-hand side above as
\begin{equation}
     \{g(Y),\, Y\in \lat{\cl}\; \text{and}\; X\subseteq Y\}=g(X^\uparrow \cap  \lat{\cl} ),
\end{equation}
where we have implicitly lifted the action of $g$ to $2^{2^{\mgs}}$ (the set of all possible collections of subsets of $\mgs$), and $X^\uparrow$ is the up-set of $X$ in $2^{\mgs}$ (which consists of all subsets of $\mgs$ containing $X$). Similar to before, this action is an isomorphism of the lattice $2^{2^{\mgs}}$ and commutes with intersection, implying
\begin{equation}
    g(X^\uparrow \cap \lat{\cl} )= g(X^\uparrow) \,\cap\, g(\lat{\cl}).
\end{equation}
Since we have assumed that $\lat{\cl}$ is closed under $g$, it follows that $g(\lat{\cl})\subseteq \lat{\cl}$. Further using the fact that $g$ is an isomorphism of $2^{2^{\mgs}}$, this inclusion is saturated, and we have
\begin{equation}
   g(X^\uparrow)\,\cap\, g(\lat{\cl}) = g(X^\uparrow)\,\cap\,\lat{\cl}.
\end{equation}
Finally, since any lattice isomorphism is also an order isomorphism, it follows
\begin{align}
     g(X^\uparrow)\,\cap\,\lat{\cl} & =  g(X)^\uparrow\,\cap\,\lat{\cl} = \{Y,\, Y\in \lat{\cl}\; \text{and}\; g(X)\subseteq Y\} ,
\end{align}
from which we obtain the desired result
\begin{align}
    g(\cl(X)) & =\bigcap\; \{Y,\, Y\in \lat{\cl}\; \text{and}\; g(X)\subseteq Y\} = \cl(g(X)).
\end{align}

\paragraph{Is the action of $G$ on $\lat{\cl}$ a lattice homomorphism?} We have just proved that with only the assumption that $\lat{\cl}$ is closed under $g$, it already follows that $\cl$ and $g$ commute. However, one of our assumptions in \S\ref{subsec:set-up} was that $g$ also preserves the lattice structure of $\lat{\cl}$, i.e., that $g$ is a lattice \textit{automorphism}. In general, this requirement is stronger than just demanding the closure of the set $\lat{\cl}$ with respect to the group action. Nevertheless, it is natural to ask for the special case where the lattice is $\lat{\cl}$, whether the closure requirement of $\lat{\cl}$ under $g$ automatically implies that $g$ is a lattice automorphism.

Consider an arbitrary collection of elements $\{X_i\}_{i\in I}$ in $\lat{\cl}$. Since the meet in $\lat{\cl}$ is given by intersection, and $g$ is an automorphism of $2^{\mgs}$,
\begin{equation}\label{eq:app-meet}
    g\bigg(\bigwedge_{i\in I}X_i\bigg) = g\bigg(\bigcap_{i\in I}X_i\bigg) = \bigcap_{i\in I}g(X_i).
\end{equation}
Under the assumption that $\lat{\cl}$ is closed under $g$, each $g(X_i)$ is also an element of $\lat{\cl}$, implying that the right-hand side of the above equation can be rewritten in terms of the meet of $\lat{\cl}$. Thus, we have
\begin{equation}
    g\bigg(\bigwedge_{i\in I}X_i\bigg) = \bigwedge_{i\in I}g(X_i),
\end{equation}
so that the meet is preserved. Indeed, this result is expected since the meet in the two lattices is the same operation. We now turn to checking whether the join is also preserved by $g$.

Since $\lat{\cl}$ is a complete lattice, we can write the join in terms of the meet to be
\begin{equation}
    \bigvee_{i\in I}X_i=\bigwedge\{Y,\, X_i\preceq Y\;\; \forall i \in I\}.
\end{equation}
Notice that this expression is structurally very similar to \eqref{eq:closure_from_int}, and by almost identical reasoning it follows that the join is preserved. In conclusion, since the lattice that we are considering inside $2^{\mgs}$ is not an arbitrary lattice, but rather the lattice of closed sets of a closure operator, even if this lattice is not a sublattice of $2^{\mgs}$, it is sufficient to assume that it is closed under the group action to obtain a lattice automorphism.

\bibliography{entropy-cone-bib}{}
\bibliographystyle{utphys}

\end{document}